\documentclass[11pt]{article}
\usepackage[margin=1in]{geometry}

\usepackage{xspace}
\usepackage{xcolor}
\usepackage{amscd}
\usepackage{amsmath}
\usepackage{amssymb}
\usepackage{amstext}
\usepackage{amsthm}
\usepackage{thmtools}
\usepackage{thm-restate}
\usepackage{bbold}
\usepackage{bm}
\usepackage{colonequals}
\usepackage{tikz}
\usepackage{booktabs}
\usepackage{array}
\usepackage{tabularx}
\usepackage{enumitem}
\usepackage[pagebackref=true,colorlinks,urlcolor=blue,linkcolor=purple,citecolor=blue,bookmarks,bookmarksopen,bookmarksnumbered]{hyperref}
\usepackage[capitalise,nameinlink]{cleveref}

\usepackage{mathtools}
\usepackage{amssymb, amsfonts, amsthm}

\usepackage{array}
\usepackage{verbatim}
\usepackage{mathrsfs}
\usepackage{graphicx}
\usepackage{algorithm}
\usepackage{algpseudocode}
\usepackage{caption}

\makeatletter
\newenvironment{breakablealgorithm}
{%
	\begin{center}
		\refstepcounter{algorithm}
		\hrule height.8pt depth0pt \kern2pt
		\renewcommand{\caption}[2][\relax]{
			{\raggedright\textbf{\fname@algorithm~\thealgorithm} ##2\par}%
			\ifx\relax##1\relax
				\addcontentsline{loa}{algorithm}{\protect\numberline{\thealgorithm}##2}%
			\else
				\addcontentsline{loa}{algorithm}{\protect\numberline{\thealgorithm}##1}%
			\fi
			\kern2pt\hrule\kern2pt
		}
		}{%
		\kern2pt\hrule\relax	\end{center}
}
\makeatother

\newcommand*{\alglineref}[1]{\hyperref[#1]{Line~\ref*{#1}}}

\usepackage{tikz}
\usetikzlibrary{decorations.pathreplacing}

\usepackage[toc,page]{appendix}

\usepackage{titlesec}
\titleformat{\subsubsection}[runin]
{\normalfont\normalsize\bfseries}
{\thesubsubsection}{1em}{}

\theoremstyle{plain}
\newtheorem{theorem}{Theorem}[section]
\newtheorem{lemma}[theorem]{Lemma}

\newtheorem{corollary}[theorem]{Corollary}
\newtheorem{assumption}[theorem]{Assumption}

\theoremstyle{definition}
\newtheorem{definition}[theorem]{Definition}

\newtheorem{problem}{Problem}

\newtheorem{claim}[theorem]{Claim}

\newtheorem{property}[theorem]{Property}
\theoremstyle{remark}
\newtheorem{remark}[theorem]{Remark}

\newcommand{\sA}{\mathcal{A}}
\newcommand{\sB}{\mathcal{B}}

\newcommand{\sD}{\mathcal{D}}
\newcommand{\sE}{\mathcal{E}}

\newcommand{\sO}{\mathcal{O}}

\newcommand{\sS}{\mathcal{S}}

\newcommand{\sU}{\mathcal{U}}

\newcommand{\NN}{\mathbb{N}}

\newcommand{\N}{\mathbb{N}}

\newcommand{\R}{\mathbb{R}}

\newcommand{\F}{\mathbb{F}}

\DeclareSymbolFont{bbold}{U}{bbold}{m}{n}
\DeclareSymbolFontAlphabet{\mathbbold}{bbold}

\DeclareMathOperator*{\E}{\mathbb{E}}

\newcommand{\Ber}{\mathrm{Ber}}

\newcommand{\poly}{\mathrm{poly}}

\newcommand{\LPN}{\mathsf{LPN}}

\newcommand{\defeq}{\coloneqq}

\newcommand{\eps}{\varepsilon}
\renewcommand{\epsilon}{\varepsilon}

\newcommand{\enc}{\mathsf{Enc}}
\newcommand{\dec}{\mathsf{Dec}}
\newcommand{\key}{\mathsf{KeyGen}}
\newcommand{\sk}{\mathsf{sk}}
\newcommand{\pk}{\mathsf{pk}}
\newcommand{\negl}{\mathsf{negl}}
\newcommand{\ed}{\Delta_E}
\newcommand{\ham}{\Delta_H}
\newcommand{\emb}{\mathsf{Emb}}

\newcommand{\prcz}{\mathsf{PRC}_0}
\newcommand{\keyz}{\mathsf{KeyGen}_0}
\newcommand{\encz}{\mathsf{Enc}_0}
\newcommand{\decz}{\mathsf{Dec}_0}
\newcommand{\pz}{p_0}

\newcommand{\ecc}{\mathsf{ECC}}

\newcommand{\prg}{\mathsf{PRG}}
\newcommand{\prc}{\mathsf{PRC}}

\newcommand{\syn}{\mathsf{syn}}

\newcommand{\inn}{\mathrm{in}}

\newcommand{\eccseed}{\ecc_{\mathrm{seed}}}
\newcommand{\pseed}{p_{\mathrm{seed}}}
\newcommand{\eccmsg}{\ecc_{\mathrm{msg}}}
\newcommand{\pmsg}{p_{\mathrm{msg}}}
\newcommand{\rmsg}{R_{\mathrm{msg}}}
\newcommand{\eccload}{\ecc_{\mathrm{load}}}
\newcommand{\pload}{p_{\mathrm{load}}}
\newcommand{\rload}{R_{\mathrm{load}}}

\newcommand{\shiftpk}{\mathsf{shift}_{\mathsf{pk}}}
\newcommand{\shiftmsg}{\mathsf{shift}_{\mathrm{msg}}}
\newcommand{\shiftsyn}{\mathsf{shift}_{\mathrm{syn}}}

\newcommand{\sigmsg}{\Sigma_{\mathrm{msg}}}

\newcommand{\sigsyn}{\Sigma_{\mathrm{syn}}}

\newcommand{\prgmsg}{\prg_{\mathrm{msg}}}
\newcommand{\prgsyn}{\prg_{\mathrm{syn}}}

\newcommand{\loc}{\mathsf{loc}}

\newcommand{\msg}{\mathrm{msg}}
\newcommand{\lcs}{\mathsf{LCS}}
\newcommand{\seed}{\mathrm{seed}}

\newif\ifshownotes
\shownotestrue

\ifshownotes
	\newcommand{\znote}[1]{\textcolor{blue}{\textbf{[Zeh: #1]}}}
	\newcommand{\songtaonote}[1]{\textcolor{red}{\textbf{[Songtao: #1]}}}
	\newcommand{\hstnote}[1]{\textcolor{olive}{\textbf{[Shengtang: #1]}}}
	\newcommand{\xnote}[1]{\textcolor{purple}{\textbf{[Xin: #1]}}}
\else
	\newcommand{\znote}[1]{}
	\newcommand{\songtaonote}[1]{}
	\newcommand{\hstnote}[1]{}
	\newcommand{\xnote}[1]{}
\fi

\title{High-Rate Public-Key Pseudorandom Codes for Edit Errors}
\author{
Shengtang Huang\thanks{\texttt{peanuttang@mail.ustc.edu.cn}, School of the Gifted Young, University of Science and Technology of China. Work done in part while visiting Johns Hopkins University.}
\and
Xin Li\thanks{\texttt{lixints@cs.jhu.edu}, Department of Computer Science, Johns Hopkins University.}
\and
Songtao Mao\thanks{\texttt{smao13@jhu.edu}, Department of Computer Science, Johns Hopkins University.}
\and
Zhaienhe Zhou\thanks{\texttt{zhaienhezhou@gmail.com}, School of the Gifted Young, University of Science and Technology of China. Work done in part while visiting Johns Hopkins University.}
}
\date{}
\begin{document}

\maketitle
\begin{abstract}
    Pseudorandom codes (PRCs), introduced by Christ and Gunn (CRYPTO '2024), are error-correcting codes whose codewords are computationally indistinguishable from uniformly random strings, while still being decodable by someone holding the key. They provide a natural primitive for robust and undetectable watermarking, particularly in applications to AI-generated content. Although recent works have obtained strong results for substitution errors, the edit-error setting remains much less understood, especially in the high-rate regime and over small alphabets.

    We study public-key pseudorandom codes against edit errors. First, we give a new reduction showing that binary zero-bit PRCs robust against a constant fraction of substitution errors can be transformed into binary zero-bit PRCs robust against edit errors. Consequently, under any assumption that yields zero-bit Hamming-robust PRCs, one also obtains zero-bit PRCs for edit channels, albeit only for the weaker class of sublinear polynomial edit channels, namely channels with edit error rate $1/n^{\gamma}$ for any constant $\gamma>0$.

    In the high-rate regime, we construct public-key PRCs with rate arbitrarily close to $1$ over sufficiently large constant alphabets, and with rate arbitrarily close to $1/2$ over the binary alphabet. Moreover, if we allow the alphabet size to be $\poly(\lambda)$, where $\lambda$ is the security parameter, then our public-key PRCs can attain the Singleton bound for insertion-deletion channels. Taken together, these results yield the first high-rate public-key binary PRC constructions for edit channels, under the same assumption that yields zero-bit Hamming-robust PRCs.
\end{abstract}
\thispagestyle{empty}
\newpage
\tableofcontents
\thispagestyle{empty}
\newpage
\setcounter{page}{1}

\section{Introduction}
The rapid spread of AI-generated content has made reliable provenance identification increasingly important. A natural approach is watermarking: embedding a hidden signal into generated content so that its origin can later be verified. For such a watermark to be useful, it should satisfy two basic requirements. First, it should provide \emph{undetectability}: even an efficient adversary making many adaptive queries cannot distinguish the watermarked model from the original one. In particular, watermarking should cause no computationally detectable degradation in output quality. Second, it should provide \emph{robustness}: given the secret key, the hidden signal should remain efficiently detectable even after the generated content has been modified. Reconciling these two requirements is the central challenge of watermarking. Starting with \cite{Aaronson_blog_AT_safety, kirchenbauer2023watermark}, a growing line of work has explored watermarking schemes for generative models, including approaches based on modifying the randomness used during generation \cite{CGZ24, zhao2024provable, kirchenbauer2024on, kuditipudi2024robust, fairoze2025publicly}. These works pursue different notions of watermarking and different trade-offs between quality and robustness, and together they highlight the difficulty of achieving both strong undetectability and strong robustness simultaneously.

A recent cryptographic approach to this problem is based on \emph{pseudorandom codes} (PRCs), introduced by Christ and Gunn \cite{christPseudorandomErrorCorrectingCodes2024}. Informally, a PRC is a keyed error-correcting code whose codewords are computationally indistinguishable from uniformly random strings to any efficient observer without the decoding key, while still being decodable by someone holding the key even after the codeword has been corrupted. In addition to robustness and pseudorandomness, PRCs satisfy a soundness property: strings unrelated to genuine codewords should decode to $\bot$ except with negligible probability. These properties make PRCs a natural primitive for building watermarks that are both robust and undetectable.

Since their introduction, pseudorandom codes have been studied from several complementary directions. On the positive side, prior works have given constructions under a range of assumptions, clarified their connections to watermarking and related cryptographic primitives, strengthened robustness guarantees in the \emph{substitution-error} setting, and developed stronger security notions such as adaptive robustness, ideal security, and CCA-style security \cite{golowich2024edit, GG25_RANDOM, AAC+24_stoc25, christ2025improved}. They have also led to concrete watermarking applications for images, video, and other generative models \cite{zhao2025sok,gunn2025an, cohen2025watermarking}. On the negative side, recent works have established strong black-box barriers, showing that PRCs tolerating a constant error rate cannot, in general, be based solely on generic cryptographic primitives \cite{garg2025black,dottling2025separating}. Together, these results suggest that PRCs form a rich and subtle primitive at the intersection of coding theory, cryptography, and watermarking.

For \emph{edit errors} (insertions, deletions, and substitutions), however, the picture is substantially more challenging. Insertions and deletions destroy synchronization, so one must cope not only with corrupted symbols but also with the loss of alignment. Despite recent progress, the theory of pseudorandom codes under edit errors remains far from satisfactory. The original work of \cite{christPseudorandomErrorCorrectingCodes2024} handled only a constant fraction of \emph{random} substitutions and deletions. Golowich and Moitra \cite{golowich2024edit} obtained a public-key PRC robust to a constant fraction of edit errors, but over a \emph{polynomial-sized alphabet}, which translates to a substantially stronger entropy requirement in watermarking applications. More recently, Christ, Golowich, Gunn, Moitra and Wichs \cite{christ2025improved} gave the first \emph{binary-alphabet} PRC robust to a constant fraction of edit errors, but only in the secret-key setting and under the less standard \emph{permuted codes conjecture}. Moreover, both \cite{golowich2024edit} and \cite{christ2025improved} construct only zero-bit PRCs in the edit-error setting. In summary, to the best of our knowledge, prior to this work there was no public-key PRC for adversarial edit channels over any constant-sized alphabet, or any compelling high-rate public-key PRC for edit channels.

From the application side, edit errors naturally model transformations that arise in modern AI pipelines, including paraphrasing, rewriting, cropping, and other operations that disrupt synchronization \cite{krishna2023paraphrasing,sadasivan2023can,qu2025provably,kuhn2024measuring,shu2024visual,awasthi2019parallel}. This is especially relevant because such transformations are often exactly the kinds of post-processing steps an adversary, editor, or downstream system may apply to generated content. Moreover, in watermarking applications, the relevant efficiency parameter is not only a quantitative concern but also affects applicability. In PRC-based watermarking, the codeword is used to drive the sampling randomness of the generative process, and detection requires recovering a noisy version of this codeword from the generated content. Thus, highly redundant or low-rate constructions may require longer or higher-entropy outputs before enough of the embedded structure can be recovered and decoded, which can significantly restrict applicability in low-entropy regimes \cite{kuditipudi2024robust, fairoze2025publicly}. Thus, constructing edit-robust PRCs with high rate is a central goal both theoretically and in practice.

In this work, we develop a general framework for high-rate public-key pseudorandom codes for edit channels over constant-sized alphabets. In particular, we give the first construction of a public-key binary zero-bit PRC with suitable edit robustness, and then use it as a building block to obtain high-rate public-key PRCs over small alphabets.

\subsection{Main Results}

We first introduce the channel notions and assumptions used in our results.

\begin{definition}
    Let $p \in (0,1)$. A channel $\sE:\Sigma^\ast \to \Sigma^\ast$ is called a \emph{$p$-bounded substitution (edit) channel} if, for every string $x \in \Sigma^n$, $\sE(x)$ differs from $x$ in at most $pn$ substitutions (respectively, is within edit distance at most $pn$ from $x$). The related $p$-bounded \emph{insertion-deletion} channel only allows at most $pn$ insertions and deletions.
    
    A channel $\sE:\Sigma^\ast \to \Sigma^\ast$ is called a \emph{sublinear polynomial substitution (edit) channel} if there exists a constant $0 < c<1$ such that, for every string $x \in \Sigma^n$, $\sE(x)$ differs from $x$ in at most $n^{1 - c}$ substitutions (respectively, is within edit distance at most $n^{1 - c}$ from $x$).
\end{definition}

\begin{assumption}[Informal version of \Cref{asp:CG24_comb}, \cite{christPseudorandomErrorCorrectingCodes2024}]\label{asp:informal}
At least one of the following holds:
\begin{itemize}
    \item \textbf{Subexponential LPN:} $\LPN$ is hard for every $2^{O(\sqrt n)}$-time adversary;
    \item \textbf{Polynomial LPN + planted XOR:} Both $\LPN$ and $\mathsf{XOR}$ are hard for every polynomial-time adversary.
\end{itemize}
\end{assumption}
We stress that this is the same assumption that yields zero-bit Hamming-robust PRCs, and in fact, for our constructions we can use any such assumption.

We begin by showing how to obtain zero-bit PRCs that are robust to edit errors from zero-bit PRCs over the binary alphabet that are robust to substitution errors. This yields a generic reduction from Hamming robustness to edit robustness, at the cost of a subpolynomial loss in the tolerable error rate.

\begin{theorem}[Reduction from Hamming to edit PRCs, informal version of \Cref{thm: CGK_construction}]\label{thm:intro-zero}
    Let $\prc_h$ be a binary zero-bit public-key (resp., secret-key) PRC that is robust against every $p$-bounded substitution channel for some constant $p>\frac{1}{4}$. Then there exists a zero-bit public-key (resp., secret-key) PRC over the binary alphabet that is robust against every sublinear polynomial edit channel.
\end{theorem}

Combining \Cref{thm:intro-zero} with the zero-bit Hamming PRC construction of \cite{christPseudorandomErrorCorrectingCodes2024} under \Cref{asp:informal}, we obtain the following corollary.

\begin{corollary}[Binary zero-bit edit PRC, informal version of \Cref{thm: edit_PRC_zero_bit}]\label{cor:intro-zero-cor}
    Under \Cref{asp:informal}, there exists a binary zero-bit public-key (secret-key) PRC that is robust against every sublinear polynomial edit channel.
\end{corollary}
We note that this gives the first edit-robust public-key zero-bit PRC over the binary alphabet. Starting from the zero-bit construction above, we then develop a general framework for boosting the rate. Our next result shows that, over a larger alphabet, one can obtain pseudorandom codes with rate arbitrarily close to $1$ while still tolerating either a constant fraction of edit errors or a sublinear polynomial number of edit errors, depending on the underlying zero-bit PRC.

\begin{theorem}[Multi-bit to zero-bit reduction, informal version of \cref{cor:multi_zero_redu_sync}]\label{thm:reduction_to_zero_bit_1}
For every sufficiently small $\varepsilon>0$ and every alphabet $\Sigma$ with $|\Sigma| \ge \poly(1/\varepsilon)$, the following holds.

If there exists a zero-bit public-key (resp., secret-key) PRC over $\Sigma$ that is robust against every $p$-bounded edit channel (resp., sublinear polynomial edit channel), then there exists a public-key (resp., secret-key) PRC over the same alphabet $\Sigma$ with rate $1-\varepsilon$ that is robust against every $p'$-bounded edit channel (resp., sublinear polynomial edit channel) for some $p'=\Omega(\min\{p,\varepsilon\})$.
\end{theorem}

We then plug in different zero-bit PRC constructions to obtain:

\begin{corollary}[PRC with rate close to $1$, informal version of \cref{cor:body-rateone-result1,cor:body-rateone-result2}]\label{thm:intro-onerate}

For every sufficiently small constant $\varepsilon>0$ and any security 
parameter $\lambda\in\N$, there exist public-key PRCs with rate $1-\varepsilon$ 
such that:
\begin{itemize}
    \item Over alphabets of size $\poly(\lambda,1/\eta)$, for any constant $0<\eta<1$, that are robust against every $(\varepsilon-\eta)$-bounded edit channel, under the same cryptographic assumptions as~\cite{golowich2024edit}
	(local weak PRFs).
    
    \item Over alphabets of size $\poly(1/\varepsilon)$ that are robust against every sublinear polynomial edit channel, under \cref{asp:informal}.
\end{itemize}
\end{corollary}

In particular, we formulate the first item for Insdel channels because in that setting the achieved tradeoff is essentially optimal: it matches the Singleton bound.

Finally, by instantiating our framework with binary zero-bit PRCs, we obtain a binary construction with positive constant rate. In particular, we achieve rate arbitrarily close to $1/2$ while retaining robustness against a sublinear polynomial number of edit errors.

\begin{theorem}[Binary PRC with rate close to $1/2$, informal version of \Cref{thm:body-binary-reduction}]\label{thm:intro-binary}
For every security parameter $\lambda>0$ and every sufficiently small $\varepsilon>0$, the following holds.

If there exists a zero-bit public-key (resp., secret-key) binary PRC that is robust against every $p$-bounded edit channel (resp., sublinear polynomial edit channel), then there exists a public-key (resp., secret-key) binary PRC with rate $1/2-\varepsilon$ that is robust against every $p'$-bounded edit channel (resp., sublinear polynomial edit channel), for some $p' = \Omega\left(\min\left\{p,\frac{\varepsilon^3}{\log(1/\varepsilon)}\right\}\right)$.
\end{theorem}

\begin{corollary}[Informal version of \cref{cor:body-binary-reduction}]\label{cor:intro-binary-cor}
Under \Cref{asp:informal}, we have binary public-key PRCs with rate $1/2-\varepsilon$ that are robust against every sublinear polynomial edit channel.
\end{corollary}

For comparison, we summarize prior PRC constructions together with our results in \Cref{tab:prc-comparison}.

\begin{table}[t!]
    \centering
    \renewcommand{\arraystretch}{1.35}
    \small
    \begin{tabular}{
    >{\centering\arraybackslash}m{2.0cm}
    >{\centering\arraybackslash}m{1.6cm}
    >{\centering\arraybackslash}m{1.9cm}
    >{\centering\arraybackslash}m{1.9cm}
    >{\centering\arraybackslash}m{2.2cm}
    >{\centering\arraybackslash}m{1.5cm}
    >{\centering\arraybackslash}m{3.2cm}}
    \toprule
    \textbf{Work} & \textbf{Rate} & \textbf{Alphabet} & \textbf{Error type} & \textbf{Error bound} & \textbf{PK/SK} & \textbf{Assumptions} \\
    \bottomrule
    \cite{christPseudorandomErrorCorrectingCodes2024}
    & constant
    & binary
    & substitution / random deletion
    & $O(n)$
    & PK / SK
    & LPN or LPN+XOR
    \\
    \midrule
    \cite{golowich2024edit}
    & zero-bit
    & $\poly(\lambda)$
    & edit
    & $O(n)$
    & PK
    & Local weak PRFs
    \\
    \midrule
    \cite{GG25_RANDOM}
    & zero-bit
    & binary
    & substitution
    & $O(n)$
    & PK
    & Planted hyperloop,
    \\
    \midrule
    \cite{GG25_RANDOM}
    & zero-bit
    & binary
    & substitution
    & $O(n)$
    & PK
    & Weak planted XOR
    \\
    \midrule
    \cite{GG25_RANDOM}
    & zero-bit
    & binary
    & substitution
    & $O(n)$
    & SK
    & Unconditional (against polytime, space-bounded adversaries)
    \\
    \midrule
    \cite{AAC+24_stoc25}
    & constant
    & binary
    & substitution
    & $O(n)$
    & PK / SK
    & LPN or LPN+XOR; random oracle for CCA-security
    \\
    \midrule
    \cite{christ2025improved}
    & constant
    & binary
    & edit
    & $O(n)$
    & SK
    & Permuted codes conjecture
    \\
    \midrule
    \midrule
    \cref{cor:intro-zero-cor}
    & zero-bit
    & binary
    & edit
    & $O(n^{1 - c})$
    & PK / SK
    & LPN or LPN+XOR
    \\
    \midrule
    \cref{thm:intro-onerate}
    & $1-\varepsilon$
    & $\poly(\lambda)$
    & edit
    & $O(n)$
    & PK / SK
    & LPN or LPN+XOR
    \\
    \midrule
    \cref{thm:intro-onerate}
    & $1-\varepsilon$
    & $\poly(1/\varepsilon)$
    & edit
    & $O(n^{1 - c})$
    & PK / SK
    & LPN or LPN+XOR
    \\
    \midrule
    \cref{cor:intro-binary-cor}
    & $\frac{1}{2}-\varepsilon$
    & binary
    & edit
    & $O(n^{1 - c})$
    & PK / SK
    & LPN or LPN+XOR
    \\
    \bottomrule
    \end{tabular}
    \caption{Comparison of prior PRC constructions and our new edit-robust PRCs. $n$ is the codeword length of PRCs, and $\lambda$ is the security parameter. Let $c \in (0, 1)$ be any constant. For the precise LPN and LPN$+$XOR assumptions, see \cref{asp:CG24_comb}.}
    \label{tab:prc-comparison}
\end{table}

\subsection{Technical Overviews}

In our later proofs, we work with insertion-deletion (Insdel) errors only, without allowing substitutions. This is without loss up to a factor of $2$, since each substitution can be simulated by one deletion and one insertion.

A key difficulty in constructing PRCs for edit errors is that one can no longer directly apply sparse parity checks, which are naturally resistant to Hamming errors and therefore have been used extensively in Hamming-robust (public-key) PRC constructions. Sparse parity checks are also closely related to linear codes and make it easy to derive the pseudorandom property from well-established hardness assumptions such as LPN and planted XOR. However, sparse parity checks break down completely in the edit-error setting, since even a single edit operation can cause shifts in the entire codeword, and thus the operation does not need to affect the exact index involved in the parity check in order to corrupt it.

One natural approach would then be to transform a Hamming-robust PRC into an edit-robust PRC, for example by attempting to attach a synchronization string \cite{haeuplerSynchronizationStringsCodes2021} to a Hamming-robust PRC. This approach has been very successful in transforming standard codes for Hamming errors to codes for edit errors. However, in the setting of PRCs, it is unclear how to choose the synchronization string. On the one hand, if we use a fixed synchronization string, then part of the codeword remains identical across different queries, which violates the pseudorandom property. On the other hand, random strings are also good synchronization strings, so one can try to use random strings here. However, then the encoder would need to share with the decoder the randomness used to generate the synchronization string, which requires an edit-robust PRC in the first place. Indeed, to overcome this, the work of \cite{golowich2024edit} needs to use a polynomial-sized alphabet to record the index, together with some extra effort to make it pseudorandom.

Constructing public-key edit-robust PRCs poses additional challenges. The construction in \cite{christ2025improved} appears difficult to adapt to the public-key setting, because it relies on randomly permuting the indices and alphabet during code generation to make it edit-robust. These permutations must be used during encoding, and therefore in the public-key setting would need to be included in the public key. However, once they are revealed to the adversary, the Permuted Code assumption used in the pseudorandomness proof no longer seems applicable.

In this paper, we present a novel approach that indeed transforms a Hamming-robust PRC into an edit-robust PRC. In addition, our transformation preserves the alphabet size and the public-key security of the Hamming-robust PRC. We describe the details below.

\paragraph{Zero-bit public-key edit PRCs over binary alphabet.} The crucial tool we use here is an \emph{embedding} from edit distance to Hamming distance. Specifically, the CGK embedding \cite{CGK_embedding} is a randomized mapping $\emb$ that transforms edit distance into Hamming distance with quadratic distortion: for strings $x, y$ with edit distance $k$, $\ham\big(\emb(x, r), \emb(y, r)\big) = O(k^2)$ with constant probability over the randomness of $r$. Furthermore, this function is reversible. Given an embedding string $e$ and the random seed $r$, we can recover $x = \emb^{-1}(e, r)$.

To construct a zero-bit edit-robust PRC, we start from a Hamming-robust $\prc_h$ (e.g., from \cite{christPseudorandomErrorCorrectingCodes2024, GG25_RANDOM}) that tolerates a $(1 / 4 + \varepsilon)$ fraction of substitution errors. The high level idea is that we want to use the CGK embedding to map our edit-robust PRC codewords to codewords in the Hamming-robust PRC. This suggests the following encoding and decoding procedure for a single-block construction:
\begin{itemize}
    \item We first generate the public/secret key pair $(\pk_h, \sk_h)$ via the key generator of $\prc_h$, and also sample a random seed $r$ in the public key for the CGK embedding.
    
    \item For encoding, first compute a Hamming codeword $a \leftarrow \enc_h(1^{\lambda}, \pk_h, 1)$, and then output an edit-robust PRC codeword $y = \emb^{-1}(a, r)$.
    
    \item For decoding, the decoder takes the received string $y'$ and computes $\hat a = \emb(y', r)$. If $\dec_h(\sk_h, \hat a) = 1$, it outputs $1$, otherwise $\bot$.
\end{itemize}

However, one issue with the above is that in the second step, not every binary string is a valid embedding under seed $r$, hence we may not be able to find such a $y = \emb^{-1}(a, r)$. To address this, we ``project'' $a$ to a valid embedding $b$ under seed $r$ and output $y = \emb^{-1}(b, r)$ instead.

Note that the CGK embedding is generated by a random walk on the input indices, and copying the bits of the input string to bits of the output string at corresponding indices. At each step of the random walk, with probability $1/2$ it stays at the same index and with probability $1/2$ it moves to the next index. Whenever the walk stays at the same index, the corresponding output bits must be consistent, because they come from the same input bit. If $a$ violates this, then $a$ is not a valid embedding, but we can repair it by flipping the conflicting bits. We show that doing this will add at most $1 / 4$ fraction of Hamming errors, 
which is still fine because the Hamming-robust PRC can tolerate any $(1/4+\eps)$ fraction of substitution errors.

For pseudorandomness, the key observation here is that for any fixed $r$, if we first ``project'' $a$ to a valid embedding $b$ under $r$ and then set $y = \emb^{-1}(b, r)$, this process actually induces an injective mapping from the indices of $y$ to those of $a$, that is, each coordinate of $y$ is determined by a distinct coordinate of $a$. Hence, when $a$ is a uniformly random string, $y$ is also distributed uniformly. Consequently, by the pseudorandomness of $\prc_h$, the Hamming code $a$ is computationally indistinguishable from uniform, and thus the resulting edit codeword $y$ is also computationally indistinguishable from uniform. Note that the above reasoning holds for any fixed $r$, so $r$ can be made public and treated as part of the public key.

The above single-block construction works whenever the CGK embedding succeeds, but the embedding succeeds only with constant probability. To amplify, we take $\ell = \poly(n)$ independent copies, where $n$ is the length of each block. Each copy uses independent keys and seeds, and the final codeword is the concatenation $y = y_1 \circ \cdots \circ y_\ell$. Edit errors may shift block boundaries, but because this is a zero-bit PRC, which has only one possible message, for decoding we only need to hit a single good block. The decoder enumerates all substrings of the received string via a sliding window, treating each substring as a candidate block. By Markov's inequality at least a constant fraction of the blocks suffer few edit errors. For such a good block the distortion bound guarantees that $\hat a = \emb(y', r)$ lies within the Hamming correction radius with constant probability. Since the blocks are independent, the probability that no good block succeeds is at most $\Theta(1)^{\Theta(\ell)} = \negl(\lambda)$.

Finally, note that for sublinear polynomial edit channels, choosing the final codeword length appropriately allows us to tolerate any polynomially bounded sublinear number of edits. Combining these ingredients yields a zero-bit public-key PRC that is robust to any sublinear polynomial edit channel, under exactly the same cryptographic assumptions required for the underlying Hamming PRC.

\paragraph{Building high-rate PRCs via randomly interleaving seed markers.}
The main challenge in turning a zero-bit PRC into a high-rate PRC is that the decoder must recover not only whether a codeword is present (which is enough for a zero-bit PRC), but also enough of the encoder’s randomness to decode the actual payload. In particular, in LPN-based constructions, merely knowing the parity-check matrix is no longer sufficient. This forces us to introduce an additional payload code, together with extra random or pseudorandom masking to hide its structure. To avoid sacrificing the rate too much, however, we can afford to use only a short truly random seed. Our solution is to protect this short seed using the zero-bit PRC, and then expand it via a pseudorandom generator so that the resulting pseudorandom bits mask the payload code and make the overall codeword look pseudorandom.

In \cite{christPseudorandomErrorCorrectingCodes2024}, this issue is addressed by using a codeword of the form
\[
\prc_{\lambda}(s)\circ\bigl(\prg(s)\oplus c\bigr),
\]
where $c$ is the payload codeword carrying the message information, $s\in\{0,1\}^{\lambda}$ is a short random seed, and $\prc_{\lambda}(s)$ is an encoding of $s$ built from a zero-bit PRC. Concretely, each bit of $s$ is represented by either a genuine zero-bit PRC codeword, indicating that the bit is $1$, or by an independent uniform string, indicating that the bit is $0$. In this way, the decoder can first recover $s$ from $\prc_{\lambda}(s)$, and then use $\prg(s)$ to unmask the payload part $\prg(s)\oplus c$. However, this basic structure is fragile under adversarial errors: the protected seed part $\prc_{\lambda}(s)$ is much shorter than the payload part; consequently, an adversary may concentrate many of the allowed errors on $\prc_{\lambda}(s)$ without violating the overall error budget, causing the recovery of $s$ to fail and thus breaking the entire decoding procedure.

To address this issue, \cite{christPseudorandomErrorCorrectingCodes2024} introduced two additional layers of randomization: a random permutation $\pi$ over the coordinates and an additive mask $z$. Together, these ensure that both the locations of the blocks in $\prc_{\lambda}(s)$ and the symbols themselves are uniformly distributed. As a result, for any fixed error channel, the corruption is spread nearly evenly across the different parts of the codeword, so that with probability $1-\negl(\lambda)$, every individual part incurs only a small fraction of the total errors.

For edit errors, however, this approach no longer applies directly, because insertions and deletions destroy synchronization. While a random permutation can easily hide the positions of the protected blocks, it becomes much harder to recover those positions once insertions and deletions have shifted the coordinates. Similarly, a global mask is difficult to remove without first recovering the correct alignment, since one must know which received symbol corresponds to which original position.

To solve these two issues, we use the following two methods: 
\begin{itemize}
    \item To locate and recover the zero-bit PRC codewords, we insert each entire zero-bit codeword into the payload string as a contiguous block, rather than partitioning it into smaller pieces. The insertion locations of the different zero-bit codewords are sampled uniformly at random, subject to the constraint that the corresponding intervals are pairwise disjoint. The decoder first searches over all candidate intervals in the received word and attempts to recover each zero-bit PRC block separately. However, this interleaving alone does not guarantee that every block can be decoded with high probability. For example, if the edit channel concentrates a constant fraction of its deletions on one interval, then any zero-bit PRC block inserted there may be heavily corrupted, causing the corresponding seed bit to be decoded incorrectly. In general, we can only expect that a constant fraction of the zero-bit PRC blocks are recovered correctly. To overcome this, we add an additional layer of binary Hamming coding: instead of embedding the seed $s$ directly, we first encode it with a binary error-correcting code and then protect the resulting bits using the zero-bit PRC blocks. This ensures that the encoder’s random seed $s$ can still be recovered even if some of the zero-bit PRC blocks fail to decode. In addition, we append a suffix of $1$’s to the seed encoding, which serves as a validation string and prevents confusion between genuine encodings of zero-bits and null codewords, that is, random strings not generated using the public key $\pk$.
    \item To prevent the adversarial channel $\sE$ from exploiting statistical information about the codeword in order to target specific parts of the encoding, we also need the final codeword distribution to be uniform and independent of the encoding procedure. To achieve this, we require that the zero-bit PRC codewords are independently and uniformly distributed over the randomness of the public key $\pk$ and the encoder. For the payload part, instead of applying a single global mask, we mask the output of $\prg(s)$ directly. Thus, once the seed $s$ is recovered, the corresponding mask can also be removed. These two components are independent, and consequently the channel cannot gain any statistical advantage from the encoding procedure when deciding how to corrupt the codeword.
\end{itemize}

After recovering $s$, the inserted $\prc_{\lambda}$ blocks can be viewed simply as additional insertions, whose total contribution is only an $o(1)$ fraction of the overall codeword length. Therefore, if the underlying channel introduces a $p$ fraction of edit errors and the payload code can tolerate a $p+\varepsilon$ fraction of edit errors for some constant $\varepsilon>0$, then the original message can still be recovered correctly.

\paragraph{High-rate code via synchronization strings.}
Given the reduction above, it remains to instantiate the payload code with a randomly encoded code that is robust to insertion-deletion errors (insdel errors for short). For this purpose, we use a standard synchronization-string-based construction of insdel codes \cite{haeuplerSynchronizationStringsCodes2021}. A key fact, noted in \cite{haeuplerSynchronizationStringsCodes2021}, is that for every sufficiently small constant $\varepsilon>0$, a uniformly random string is $\varepsilon$-self-matching with probability $1-\exp(-\Omega(n))$. This is exactly the level of reliability we need, since in our application the decoding guarantee must hold with probability $1-\negl(n)$.

To make the construction both uniform and pseudorandom, we split the seed $s$ into two parts, say $s=(s_1,s_2)$. The first part is used to generate the synchronization string through $\prgsyn(s_1)$, while the second part is used to generate the masking string through $\prgmsg(s_2)$, which hides the payload codeword and makes the overall encoding pseudorandom. In addition, we apply public shift vectors to both components. Thus, conditioned on the public shifts, each component is uniformly distributed, while the decoder can still recover them once the seed is reconstructed.

The role of the synchronization string is to convert insdel errors into ordinary symbol corruptions after alignment. Concretely, suppose $\sigma'$ is the synchronization string embedded in the codeword, and $\tilde{\sigma}'$ is its corrupted version after at most $k$ insdel errors. If $\sigma'$ is $\varepsilon$-self-matching, then the synchronization algorithm of \cite{haeuplerSynchronizationStringsCodes2021} produces an approximate matching between positions of $\sigma'$ and positions of $\tilde{\sigma}'$. After placing the matched symbols into their intended coordinates and treating unmatched positions as erasures, the misalignment introduces at most $k+3\sqrt{\varepsilon}n$ corrupted or erased positions in total. In other words, once synchronization succeeds, the insdel channel is reduced to a Hamming-type error pattern with only an additional $3\sqrt{\varepsilon}n$ loss beyond the original $k$ edits.

Finally, we use a standard algebraic-geometry code as the base code in the synchronization-string construction. We choose it to have sufficiently high rate and a sufficiently large alphabet so that appending the synchronization string incurs only a negligible loss in the overall rate.

\paragraph{High-rate binary code via concatenation with random linear inner code.}
The approach with synchronization string discussed above no longer works well in the binary setting, since the synchronization needs a sufficiently large (albeit constant-sized) alphabet. More generally, in the theory of binary insdel codes, one typically relies on concatenation in order to obtain efficiently decodable constructions. To address the synchronization problem, many prior works insert long runs of $0$'s between adjacent inner codewords, using these buffers as additional markers for the decoder, such as the Schulman-Zuckerman code \cite{schulman2002asymptotically} and its subsequent improvements \cite{guruswami2016efficiently,guruswami2017deletion}. Another approach is to append additional redundancy derived from a document exchange protocol, as in \cite{cheng2018deterministic}. However, such additional markers are not only easy for the decoder to identify, but also easy for an adversarial distinguisher to detect, and therefore they destroy pseudorandomness. In our construction, we use the insdel binary code in \cite{cheng2023linear}, which does not rely on any additional markers. Instead, the construction uses different inner codes for different blocks. For our purposes, these inner codes do not need to be fixed in advance to satisfy any particular structural property; instead, independently sampled random linear codes already suffice. Moreover, since the randomness specifying the inner codes is generated as $\prg(s)\oplus z$, where $z$ is the public shifting vector, the resulting codeword is naturally uniformly distributed over the randomness of $z$.

The only issue with this construction is that, in order to maintain both high rate and efficient decoding, the inner codeword length is only $\Theta(\log n)$, since the decoding guarantee for each inner code is obtained by brute force over logarithmic-length blocks. As a result, randomly chosen inner codes are good with probability only $1-1/\poly(n)$, rather than with probability $1-\negl(n)$. To overcome this, we introduce an additional layer of concatenation, using $n^\beta$ such seeded codewords as building blocks for some $0<\beta<1$. This amplifies the probability that the overall construction is good from $1-1/\poly(n)$ to $1-\negl(n)$. The decoder still proceeds hierarchically, from the inner level to the outer level, and by choosing the parameters appropriately, we ensure that the final decoding succeeds with required probability.

\subsection{Open Problems and Discussions}

Our reduction based on the CGK embedding has the advantage that it does not require any additional cryptographic assumptions beyond those already needed for the underlying Hamming-robust PRC. However, it only yields robustness against a sublinear polynomial fraction of edit errors, because the CGK embedding maps $k$ edit errors to roughly $k^2$ Hamming errors. This quadratic loss is inherent to the CGK approach and cannot, in general, be improved, even in average-case settings. It is therefore an important open problem to find a different reduction from edit robustness to Hamming robustness with substantially smaller distortion. Such a reduction could plausibly lead to public-key PRCs of constant rate against constant-fraction edit channels.

Another important direction is to identify compelling computational assumptions under which one can construct public-key pseudorandom codes that are directly robust to edit errors. At present, known constructions of edit-robust PRCs either proceed by reducing edit robustness to Hamming robustness, as in our work and \cite{golowich2024edit}, or else obtain edit robustness only in the secret-key setting under less standard assumptions, such as in \cite{christ2025improved}. It is therefore a natural open problem to formulate and justify assumptions that directly support public-key edit-robust PRCs. Such assumptions could lead to conceptually simpler constructions and potentially to stronger guarantees, such as constant rate over small alphabets against a constant fraction of edit errors.

\subsection{Paper Organization}

In \Cref{sec:preliminary}, we review the necessary preliminaries. In \Cref{sec:CGK_edit_PRC}, we apply the CGK embedding to construct a binary zero-bit public-key PRC that is robust against every sublinear polynomial edit channel, thereby proving \Cref{thm:intro-zero}. In \Cref{sec:generic_reduction_to_zero_bit}, we develop a general black-box framework for transforming a zero-bit PRC into a higher-rate PRC. In \Cref{sec:rate-one}, we instantiate this framework using a synchronization-string-based payload code to prove \Cref{thm:reduction_to_zero_bit_1}. Finally, in \Cref{sec:binary}, we combine the framework with our binary concatenation-based payload code to prove \Cref{thm:intro-binary}.

\section{Preliminaries}\label{sec:preliminary}

For an integer $n\in\N$, we write $[n]\defeq\{1,2,\dots,n\}$.
For a string $x\in\Sigma^{\ast}$ over an alphabet $\Sigma$, let $|x|$ denote its length.
For $i\in[|x|]$, let $x_i$ be the $i$-th symbol.
For an interval $I = [i , j]\ (1\leq i\leq j\leq |x|)$, let $x_I$ or $x_{[i:j]}$ denote the substring $x_i x_{i+1}\cdots x_j$.
For strings $x,y$, we write $x\circ y$ for concatenation.
We write $U_n(\Sigma)$ for the uniform distribution over $\Sigma^n$. When the alphabet $\Sigma$ is clear in the context, we may abbreviate it as $U_n$.
Unless otherwise stated, all algorithms are probabilistic polynomial time (PPT).
We write $y \gets \sA(x)$ for running a randomized algorithm $\sA$ on input $x$ and outputting $y$.
We write $x \gets \Omega$ or $x \gets \sD$ for sampling $x$ uniformly from a finite set $\Omega$ or from a distribution $\sD$, respectively.

A function $\negl:\N\to\R_{\ge 0}$ is \emph{negligible} if for every constant $c>0$ there exists $N$ such that
$\negl(n) \leq n^{-c}$ for all $n\ge N$.

\begin{lemma}[Hoeffding's inequality]\label{lem:Hoeffding_ineq}
	Let $X_1, \ldots, X_n$ be independent random variables such that $X_i \in [a_i, b_i]$. Let $X = X_1 + \cdots + X_n$ and $\mu = \E[X]$. Then for any $t > 0$, we have
	$$
		\Pr[|X - \mu| \ge t] \le 2 \exp\left( -\frac{2 t^2}{\sum_{i = 1}^n (b_i - a_i)^2} \right).
	$$
\end{lemma}

\subsection{Metric and Channel Models}

We measure errors using either Hamming distance or edit distance.
\begin{itemize}
	\item For two strings $x, y \in \Sigma^n$ of equal length, the Hamming distance $\ham(x,y)$ is the number of positions at which they differ.

	\item For two strings $x \in \Sigma^n$ and $y \in \Sigma^*$ of possibly different lengths, the edit distance $\ed(x,y)  = |x| + |y| - 2 \cdot \lcs(x, y)$ is the minimum number of insertions and deletions required to transform $x$ into $y$, where $\lcs(x, y)$ is the length of the longest common subsequence between $x$ and $y$.

    This definition differs slightly from the classical edit distance: substitutions are not allowed, or equivalently, each substitution is counted as two operations. This formulation is chosen for technical convenience in our analysis and proofs.
\end{itemize}

\begin{definition}[$p(\cdot)$-bounded channels]
	Let $p : \N \to [0,1]$ be a function.
	\begin{itemize}
		\item A channel $\sE : \Sigma^* \to \Sigma^*$ is called a \emph{$p(\cdot)$-bounded Hamming channel} if for every input $x \in \Sigma^*$, the output $x' = \sE(x)$ satisfies $|x'| = |x|$ and
          $$
              \ham(x, x') \le p(|x|) \cdot |x|.
          $$

		\item A channel $\sE : \Sigma^* \to \Sigma^*$ is called a \emph{$p(\cdot)$-bounded edit channel} if for every input $x \in \Sigma^*$, the output $x' = \sE(x)$ satisfies
          $$
              \ed(x, x') \le p(|x|) \cdot |x|.
          $$
	\end{itemize}

    In particular, when $p : \N \to [0, 1]$ is a constant function, we may abbreviate it as the $p$-bounded Hamming/edit channel. And when $p(x) = c / x^\gamma$ for some constant $c > 0$ and $0 < \gamma < 1$, we also call such channel as a sublinear polynomial Hamming/edit channel.
\end{definition}
In this paper, all channels are assumed to be oblivious, namely, their error patterns are chosen independently of the sampled keys and codewords.

\subsection{Pseudorandom Code}

\begin{definition}[Secret-key PRC]
	Let $\Sigma$ be a fixed alphabet and let $\sE:\Sigma^{\ast}\to\Sigma^{\ast}$ be a channel.
	A secret-key pseudorandom error-correcting code (secret-key PRC) with robustness to $\sE$ is a triple of randomized polynomial-time algorithms
	$(\key,\enc,\dec)$,
	together with functions $\ell,n,k:\NN\to\NN$, such that for every $\lambda\in\NN$, $\key(1^\lambda)\in\{0,1\}^{\ell(\lambda)}$, $\enc:\{1^\lambda\}\times\{0,1\}^{\ell(\lambda)}\times \Sigma^{k(\lambda)}\to \Sigma^{n(\lambda)}$, $\dec:\{1^\lambda\}\times\{0,1\}^{\ell(\lambda)}\times \Sigma^{\ast}\to \Sigma^{k(\lambda)}\cup\{\bot\}$,
	and the following conditions hold:

	\begin{itemize}
		\item \textbf{Robustness.}
		      For every $\lambda\in\NN$ and every $m\in\Sigma^{k(\lambda)}$,
		      $$
			      \Pr_{\sk\gets \key(1^\lambda)}\Big[
				      \dec(1^\lambda,\sk,\sE(x))=m :x\gets \enc(1^\lambda,\sk,m)
				      \Big] \ge 1-\negl(\lambda).
		      $$

		\item \textbf{Soundness.}
		      For every fixed $c\in\Sigma^{\ast}$,
		      $$
			      \Pr_{\sk\gets \key(1^\lambda)}\big[\dec(1^\lambda,\sk,c)=\bot\big]
			      \ge 1-\negl(\lambda).
		      $$

		\item \textbf{Pseudorandomness.}
		      For any probabilistic polynomial-time adversary $\sA$,
		      $$
			      \left|
			      \Pr_{\sk\gets \key(1^\lambda)}\big[\sA^{\enc(1^\lambda,\sk,\cdot)}(1^\lambda)=1\big]
			      -
			      \Pr_\sU\big[\sA^{\sU}(1^\lambda)=1\big]
			      \right|
			      \leq \negl(\lambda),
		      $$
		      where $\sU$ is an oracle that answers each query with an independent uniform sample from $\Sigma^{n(\lambda)}$.
	\end{itemize}
\end{definition}

\begin{definition}[Public-key PRC]
	Let $\Sigma$ be a fixed alphabet and let $\sE:\Sigma^{\ast}\to\Sigma^{\ast}$ be a channel.
	A public-key pseudorandom error-correcting code (public-key PRC) with robustness to $\sE$
	is a triple of randomized polynomial-time algorithms
	$(\key,\enc,\dec)$,
	together with functions $\ell_{\pk},\ell_{\sk},n,k:\NN\to\NN$, such that for every $\lambda\in\NN$,
	$$
		(\pk,\sk)\gets \key(1^\lambda)\in \{0,1\}^{\ell_{\pk}(\lambda)}\times \{0,1\}^{\ell_{\sk}(\lambda)},
	$$
	$$
		\enc:\{1^\lambda\}\times\{0,1\}^{\ell_{\pk}(\lambda)}\times \Sigma^{k(\lambda)}\to \Sigma^{n(\lambda)},
		\qquad
		\dec:\{1^\lambda\}\times\{0,1\}^{\ell_{\sk}(\lambda)}\times \Sigma^{\ast}\to \Sigma^{k(\lambda)}\cup\{\bot\},
	$$
	and the following conditions hold:

	\begin{itemize}
		\item \textbf{Robustness.}
		      For every $\lambda\in\NN$ and every $m\in\Sigma^{k(\lambda)}$,
		      $$
			      \Pr_{(\pk,\sk)\gets \key(1^\lambda)}\Big[
				      \dec(1^\lambda,\sk,\sE(x))=m : x\gets \enc(1^\lambda,\pk,m)
				      \Big] \ge 1-\negl(\lambda).
		      $$

		\item \textbf{Soundness.}
		      For every fixed $c\in\Sigma^{\ast}$,
		      $$
			      \Pr_{(\pk,\sk)\gets \key(1^\lambda)}\big[\dec(1^\lambda,\sk,c)=\bot\big]
			      \ge 1-\negl(\lambda).
		      $$

		\item \textbf{Pseudorandomness.}
		      For any probabilistic polynomial-time adversary $\sA$,
		      $$
			      \left|
			      \Pr_{(\pk,\sk)\gets \key(1^\lambda)}\big[\sA^{\enc(1^\lambda,\pk,\cdot)}(1^\lambda,\pk)=1\big]
			      -
			      \Pr_{\substack{(\pk,\sk)\gets \key(1^\lambda) \\ \sU}}\big[\sA^{\sU}(1^\lambda,\pk)=1\big]
			      \right|
			      \leq \negl(\lambda),
		      $$
		      where $\sU$ is an oracle that answers each query with an independent uniform sample from $\Sigma^{n(\lambda)}$.
	\end{itemize}
\end{definition}

The \emph{codeword length} of a (secret-key or public-key) PRC is $n(\lambda)$ and the \emph{message length} is $k(\lambda)$. The \emph{rate} is the function $R(\lambda) := k(\lambda) / n(\lambda)$. We often drop the dependence on $\lambda$ when it is clear from context.

For both secret-key and public-key PRCs, if there is only one possible message  (i.e. $k(\lambda) = 0$), we denote this unique message by $1$ and say that the scheme is a zero-bit PRC.

\subsection{Cryptographic Assumptions}

We state the cryptographic assumptions used in \cite{christPseudorandomErrorCorrectingCodes2024} to construct PRCs robust to a constant fraction of substitution errors. Our edit PRC construction treats such Hamming PRCs as a black box. Thus, any construction of a Hamming-robust PRC, whether based on the assumptions below or on alternative assumptions (e.g., those in \cite{GG25_RANDOM}), can be seamlessly integrated. No additional assumptions are required beyond those needed for the underlying Hamming PRC.

\begin{problem}[Learning Parity with Noise $\LPN_{g,\eta}$]\label{prob:LPN}
Let $\eta\in(0,1/2)$ be a constant and let $g:\N\to\N$ be a function. Distinguish between the following two distributions over pairs in $\F_2^{n\times g(n)}\times \F_2^n$:
\begin{enumerate}
    \item \textbf{Noisy parity distribution:} $\sD_0(n,g(n),\eta)$, obtained by sampling $A\gets \F_2^{n\times g(n)}$, $s\gets \F_2^{g(n)}$, $e\gets \Ber(n,\eta)$, and outputting $(A,As\oplus e)$, where $\Ber(n,\eta)$ denotes the product Bernoulli distribution on $\F_2^n$ with parameter $\eta$.
    \item \textbf{Uniform distribution:} $\sD_1(n,g(n))$, obtained by sampling
    $A\gets \F_2^{n\times g(n)}$, $u\gets \F_2^n$,
    and outputting $(A,u)$.
\end{enumerate}
\end{problem}
\begin{problem}[Planted XOR $\mathsf{XOR}_{m,t}$]\label{prob:XOR}
Let $m,t:\N\to\N$ be functions with $m(n)=n^{\Omega(1)}$ and $t(n)=\Theta(\log n)$. Distinguish between the following two distributions over matrices in $\F_2^{n\times m(n)}$:
\begin{enumerate}
    \item \textbf{Null distribution:} $\sD_0(n,m(n))$, the uniform distribution over $\F_2^{n\times m(n)}$;
    \item \textbf{Planted distribution:} $\sD_1(n,m(n),t(n))$, obtained by first sampling a random $t(n)$-sparse vector $s\in \sS_{t(n),n}$ and then sampling a random matrix $G\in \F_2^{n\times m(n)}$ subject to $s^\top G = 0$, where $\sS_{t(n),n}\subseteq \F_2^n$ denotes the set of all vectors of Hamming weight exactly $t(n)$.
\end{enumerate}
\end{problem}

\begin{assumption}\label{asp:CG24_comb}
At least one of the following holds:

\begin{enumerate}
    \item \textbf{Subexponential hardness of $\LPN_{g,\eta}$.}
    There exist a constant $\eta\in(0,1/2)$ and a function $g:\N\to\N$ with $g(n)=\Omega(\log^2 n)$ such that no probabilistic adversary running in time $2^{O(\sqrt n)}$ can solve \Cref{prob:LPN} with non-negligible advantage.

    \item \textbf{Polynomial hardness of $\LPN_{g,\eta}$ together with hardness of $\mathsf{XOR}_{m,t}$.}
    There exist constants $\eta\in(0,1/2)$ and $\varepsilon\in(0,1)$ such that:
    \begin{itemize}
        \item No probabilistic polynomial-time adversary can solve \Cref{prob:LPN} for $g(n)=n^\varepsilon$ with non-negligible advantage.
        
        \item No probabilistic polynomial-time adversary can solve \Cref{prob:XOR} for $m(n)=2n^\varepsilon$ and $t(n)=\Theta(\log n)$ with advantage $1-\negl(n)$.
    \end{itemize}
\end{enumerate}
\end{assumption}

\begin{theorem}[{\cite[Theorem 1]{christPseudorandomErrorCorrectingCodes2024}}]\label{thm:CG24_PRC}
	Let $p \in (0,1/2)$ be any constant. Under \Cref{asp:CG24_comb}, there exists a zero-bit public-key (secret-key) pseudorandom code $\prc = (\key, \enc, \dec)$ against any $p$-bounded Hamming channel.
\end{theorem}

\subsection{Pseudorandom Generators}

\begin{lemma}[Cryptographic pseudorandom generator]\label{lem:crypto_PRG}
    Suppose that one-way functions exist. Then for any polynomial $\ell(\lambda) > \lambda$, there exists a deterministic algorithm $\prg : \{0,1\}^* \to \{0,1\}^*$ such that:
    \begin{itemize}
        \item $|\prg(s)| = \ell(|s|)$ for any input $s \in \{0, 1\}^*$.

        \item For any $\lambda \in \N$ and every PPT distinguisher $\sA$, we have
        $$
        \left| \Pr_{s \gets \{0,1\}^\lambda}\big[\sA(1^\lambda, \prg(s)) = 1\big] - \Pr_{r \gets \{0,1\}^{\ell(\lambda)}}\big[\sA(1^\lambda, r) = 1\big] \right| \le \negl(\lambda).
        $$
    \end{itemize}
\end{lemma}
Throughout this paper, we fix such a pseudorandom generator $\prg$ and use it directly whenever needed, without mentioning it explicitly each time.

\section{Zero-Bit Public-Key Edit PRCs}\label{sec:CGK_edit_PRC}

In \cite{golowich2024edit}, the authors give an edit-robust PRC construction over a large alphabet, reducing edit PRC to a Hamming PRC via indexing. More recently, in \cite{christ2025improved}, the authors give binary, constant-rate secret-key pseudorandom codes that tolerate a constant fraction of worst-case edits, based on permuted codes conjecture. In this section, we present a binary-alphabet zero-bit \emph{public-key} PRC resistant to \emph{any sublinear polynomial} edit errors, and \emph{do not require extra assumptions} other than those required by a Hamming PRC. More precisely, we give a reduction from zero-bit edit PRC to zero-bit Hamming PRC:

\begin{restatable}{theorem}{CGKconstruction}\label{thm: CGK_construction}
    Let $\prc_h = (\key_h, \enc_h, \dec_h)$ be a zero-bit public-key (secret-key) PRC against any $\left( \frac{1}{4} + c_1 \right)$-bounded Hamming channel, where $c_1 < 1 / 4$ is a fixed constant. Then for any constant $0 < \gamma < 1 / 2$, \Cref{alg: cgk} constructs a zero-bit public-key (secret-key) PRC against any $p(\cdot)$-bounded edit channel for small enough constant $c_2$ and $p : x \mapsto c_2 / x^\gamma$.
\end{restatable}

Combining \Cref{thm:CG24_PRC}, i.e. the constructions of Hamming PRCs in \cite{christPseudorandomErrorCorrectingCodes2024} under \Cref{asp:CG24_comb}, we obtain that:

\begin{theorem}\label{thm: edit_PRC_zero_bit}
    Under \Cref{asp:CG24_comb}, for any $0 < \gamma < 1 / 2$, there exists a zero-bit public-key (secret-key) PRC against any $p(\cdot)$-bounded edit channel for small enough constant $c$ and $p : x \mapsto c / x^\gamma$.
\end{theorem}

The crucial tool we use here is the CGK embedding \cite{CGK_embedding}. We will give the description and properties of the CGK embedding and decoding functions in \Cref{sec:CGK_emb}. And then show the detailed constructions of zero-bit public-key edit PRCs in \Cref{sec:zero_bit_PRC_CGK_construction}. Finally, we prove \Cref{thm: CGK_construction} in \Cref{sec:zero_bit_PRC_CGK_analysis}.

\subsection{CGK Embedding}\label{sec:CGK_emb}

The CGK embedding, originally introduced by Chakraborty, Goldenberg and Kouck\'{y} \cite{CGK_embedding}, provides a mapping from strings under edit distance to strings under Hamming distance. At a high level, the embedding performs a random walk on the index axis of the input string, which induces a partial synchronization between two input strings that are close in edit distance.

This synchronization ensures that corresponding symbols in the two strings tend to be aligned in the embedding, despite insertions and deletions.  
As a result, the embedding maps edit distance to Hamming distance with only a quadratic distortion: two strings with edit distance $k$ are mapped to embedded strings whose Hamming distance is $O(k^2)$ with constant probability.

\paragraph{CGK embedding.}
We briefly recall the idea of the CGK embedding.
The embedding maps a binary string $x \in \{0, 1\}^L$ to a longer binary string by performing a randomized walk over the indices of $x$. At each step $t$, the current symbol $x_i$ is written to the output, and the index $i$ is advanced according to a random function $h_t(x_i) \in \{0, 1\}$ derived from the shared randomness $r$. As a result, each input bit may be repeated multiple times in the output sequence, while the relative order of symbols is preserved.  
This randomized walk ensures that nearby symbols in the original string tend to remain close in the embedding.

\begin{breakablealgorithm}
\caption{CGK Embedding Function $\emb$}
\label{alg: cgk_emb}
\begin{algorithmic}[1]
\Require $L$, $x \in \{0, 1\}^L$, and $r \in \{0, 1\}^{3 L}$.
\Function{$\emb$}{$x, r$}
\State Interpret $r$ as a description of $h_1, \ldots, h_{1.5 L} : \{0, 1\} \to \{0, 1\}$ 
\State Initialize $i = 1$ and $a_{1 \sim 1.5 L} = \bot^{1.5 L}$ \Comment{$a_{1 \sim 1.5 L}$ is the resulting string}
\For{$t = 1, 2, \ldots, 1.5 L$}
    \If{$i \le L$}
        \State $a_t = x_i$ 
        \State $i = i + h_t(x_i)$ \Comment{Update the counter $i$ according to $r$}
    \Else
        \State $a_t = 0$ \Comment{Pad $a$ with $0$ if the counter $i$ overflows}
    \EndIf
\EndFor
\State \Return $a$ 
\EndFunction
\end{algorithmic}
\end{breakablealgorithm}

The CGK embedding function has a nice property: it can guarantee a quadratic distortion with constant probability over the randomness of $r$:

\begin{lemma}[Modified from {\cite[Theorem 4]{CGK_embedding}}]\label{lem: CGK_emb_prop}
    Let $\emb : \{0, 1\}^L \times \{0, 1\}^{3 L} \to \{0, 1\}^{1.5 L}$ be the embedding mapping computed by \Cref{alg: cgk_emb}. Then for every positive constant $c$ and every $x, y \in \{0, 1\}^L$, $\ham(\emb(x, r), \emb(y, r)) \le c \cdot (\ed(x, y))^2$ with probability at least $1 - 6 / \sqrt{c}$ over the randomness of $r$.
\end{lemma}

\begin{remark}
    (1) The original proof in \cite{CGK_embedding} is for the case that allows substitutions, and the bound is $1 - 12 / \sqrt{c}$. Since not allowing substitutions will at most double the edit distance, it leads to the bound $1 - 6 / \sqrt{c}$.

    (2) In the original construction of \cite{CGK_embedding}, the embedding function outputs a string of length $3L$. In our setting, we truncate the process, resulting in an output length of $1.5L$.
    
    The distortion guarantee in \Cref{lem: CGK_emb_prop} is in fact independent of the output length. Inspecting the proof in \cite{CGK_embedding}, one can see that the bound $\ham(\emb(x,r), \emb(y,r)) \le c \cdot (\ed(x,y))^2$ with probability at least $1 - 6 / \sqrt{c}$ continues to hold for this shortened version without any loss in parameters.

    While this lemma theoretically holds even if the number of random walk steps (i.e., the output length) is arbitrarily reduced, our overall PRC construction strictly requires the output length to be $\Theta(L)$. Since the channel introduces $O(\sqrt{n})$ edit errors on each block, the CGK embedding amplifies them quadratically into $O(n)$ Hamming errors. To successfully correct these errors, the underlying Hamming PRC must have a codeword length of at least $\Theta(n)$.
\end{remark}

\paragraph{CGK decoding.} We now describe how to interpret an arbitrary binary string under the CGK embedding with a fixed seed $r$.  
Not every binary string is a valid output of the CGK embedding.  
Nevertheless, for any $a \in \{0, 1\}^{1.5 L}$, one can deterministically ``project'' $a$ to a valid embedding string $b$ by enforcing the local constraints induced by $r$. This projection simulates a legal CGK walk and ensures consistency between repeated visits to the same input index.

\begin{breakablealgorithm}
\caption{CGK Decoding Function $\emb^{-1}$}
\label{alg: cgk_dec}
\begin{algorithmic}[1]
\Require $L$, $a \in \{0, 1\}^{1.5 L}$, and $r \in \{0, 1\}^{3 L}$.
\Function{$\emb^{-1}$}{$a, r$}
\State Interpret $r$ as a description of $h_1, \ldots, h_{1.5 L} : \{0, 1\} \to \{0, 1\}$ 
\State Initialize $i = 1$, $x_{1 \sim L} = \bot^L$ and $b_{1 \sim 1.5 L} = \bot^{1.5 L}$ 
\State \Comment{$x_{1 \sim L}$ is the resulting string, and $b_{1 \sim 1.5 L}$ is a valid embedding projected from $a_{1 \sim 1.5 L}$}
\For{$t = 1, 2, \ldots, 1.5 L$}
    \If{$i \le L$}
        \If{$x_i \ne \bot$ and $x_i \ne a_t$} \Comment{$r$ indicates a ``stay'' step}
            \State $b_t = 1 - a_t$ \Comment{Flip the current bit if it differs from the previous one} \phantomsection\label{line:type1_stay_collisions}
        \Else
            \State $b_t = a_t$ 
        \EndIf
        \State $x_i = b_t$ 
        \State $i = i + h_t(x_i)$ \Comment{Update the counter $i$ according to $r$}
    \Else
        \State $b_t = 0$ \Comment{Pad $b$ with $0$ if the counter $i$ overflows} \phantomsection\label{line:type2_padding_after_overflow}
    \EndIf
\EndFor
\State Fill the remaining positions $x_k$ where $x_k = \bot$ with i.i.d. bits
\State \Return $x$ 
\EndFunction
\end{algorithmic}
\end{breakablealgorithm}

\begin{lemma}\label{lem: CGK_dec_props}
    Let $\emb^{-1} : \{0, 1\}^{1.5 L} \times \{0, 1\}^{3 L} \to \{0, 1\}^L$ be the decoding function computed by \Cref{alg: cgk_dec}, and for any $a \in \{0, 1\}^{1.5L}$, let $b_{1 \sim 1.5 L}$ be the string defined and computed in the execution of $\emb^{-1}(a, r)$. Then the following statements hold:
    \begin{enumerate}
        \item Letting $x = \emb^{-1}(a, r)$, we have $\emb(x, r) = b$. Moreover, for any constant $c > 0$,
        $$
            \ham(a, b) \le \left( \frac{3}{8} + c \right) L
        $$
        with probability at least $1 - \negl(L)$ over the uniform randomness of $r$.
    
        \item For any fixed $r$, if the input $a$ is sampled uniformly from $\{0, 1\}^{1.5 L}$, then
        $$
            \emb^{-1}(a, r) \equiv U_L.
        $$
    \end{enumerate}
\end{lemma}

\begin{proof}
    (1) It is immediate from the construction of \Cref{alg: cgk_dec} that if $x = \emb^{-1}(a, r)$, then the algorithm reconstructs the valid embedding string $b = \emb(x, r)$.
    
    The Hamming distance $\ham(a, b)$ may arise from two types of positions.
    
    \medskip
    \textit{Type I (stay collisions).}
    When $i \le L$ and $x_i$ has already been assigned, if $a_t \neq x_i$, then the decoder sets $b_t = 1 - a_t$ in \alglineref{line:type1_stay_collisions}. This event happens only if $h_t(0) = h_t(1) = 0$, which occurs with probability $1/4$. Hence, $\#\{t : b_t = 1 - a_t,\, i \le L\}$ is bounded by the sum of $1.5 L$ independent Bernoulli$(1 / 4)$ random variables, with expectation $3 L / 8$. By Hoeffding's inequality, for any constant $c > 0$,
    $$
    \Pr\left[ \#\{t : b_t = 1 - a_t,\, i \le L\} > \frac{3 L}{8} + c L \right] \le 2 \exp(-2 c^2 L) = \negl(L).
    $$
    
    \medskip
    \textit{Type II (padding after overflow).}
    When $i > L$, the decoder sets $b_t = 0$ in \alglineref{line:type2_padding_after_overflow}. Since the embedding length is $1.5 L$, and the counter $i$ increases with probability $1 / 2$ at each step while $i \le L$, we have that the value $i_{\text{final}}$ of the counter $i$ after the main loop satisfies
    $$
        \E[i_{\text{final}}] = 1 + \frac{1.5 L}{2} = 1 + 0.75 L < L.
    $$
    By Hoeffding's inequality, with probability at least $1 - \negl(L)$, the counter $i$ never exceeds $L$ throughout the $1.5 L$ steps, and therefore no \textit{Type II} error occurs.
    
    \medskip
    Combining the two cases, we conclude that
    $$
        \Pr\left[ \ham(a, b) \le \left(\frac{3}{8} + c \right) L \right] \ge 1 - \negl(L).
    $$

    \medskip
    (2) Fix any $r \in \{0,1\}^{3 L}$. During the execution of \Cref{alg: cgk_dec}, let $K \subseteq [L]$ be the set of indices $j$ such that $x_j$ is assigned a value during the main loop. For each $j \in K$, $x_j$ is assigned for the first time at some unique time step $t_j \in [1.5 L]$ just after the counter $i$ increases from $j - 1$ to $j$. Since $x_j = \bot$ at this exact moment, the algorithm directly sets $x_j = a_{t_j}$. Thus, the decoding process induces an injective mapping $\pi : K \to [1.5 L]$, where $\pi(j) = t_j$ for all $j \in K$, such that each $x_j$ for $j \in K$ is exactly determined by $a_{\pi(j)}$.

    If $a$ is sampled from $U_{1.5 L}$, then the bits $\{a_{\pi(j)}\}_{j \in K}$ are mutually independent and uniformly random. It follows that the assigned bits $\{x_j\}_{j \in K}$ are i.i.d. bits. For the remaining coordinates $j \in [L] \setminus K$, $x_j$ remains $\bot$ during the loop and is filled with fresh i.i.d. bits in the final step of the algorithm. Therefore, the entire string $x$ consists of mutually independent, uniformly random bits, which means
    $$
    \emb^{-1}(U_{1.5 L}, r) \equiv U_L.
    $$
\end{proof}

\subsection{Detailed Construction}\label{sec:zero_bit_PRC_CGK_construction}

We present the zero-bit (public-key) edit PRC construction as follows.

\begin{itemize}
    \item \textbf{Key Generation:}  
    Sample $\ell = \poly(n)$ independent pairs of keys $(\pk_t, \sk_t)_{t = 1}^\ell$ of the underlying Hamming PRC $\prc_h$. Independently sample $\ell$ random seeds $r_1, \ldots, r_\ell$ for the CGK embedding, and $\ell$ independent uniform masks $e_1, \ldots, e_\ell$ for one-time padding.
    The final public key is $\pk = (\pk_t, r_t, e_t)_{t = 1}^\ell$, and the final secret key is $\sk = (\sk_t, r_t, e_t)_{t = 1}^\ell$.  
    Multiple independent seeds are used to amplify the constant success probability guaranteed by the CGK embedding to high probability in decoding.

    \item \textbf{Encoding:}  
    Since this is a zero-bit PRC, the only message is $1$. For each $t \in [\ell]$, first generate a Hamming PRC codeword $a_t \gets \enc_h(1^\lambda, \pk_t, 1)$, and mask it with the one-time pad to get $a'_t = a_t \oplus e_t$. Then project $a'_t$ to a valid CGK embedding $b_t$ under seed $r_t$, and output $y_t = \emb^{-1}(b_t, r_t) = \emb^{-1}(a'_t, r_t)$.  
    The final codeword is $y = y_1 \circ \cdots \circ y_\ell$.

    \item \textbf{Decoding:}
    Given a received string $y' \in \{0, 1\}^*$, the decoder enumerates all substrings $z = y'_{[i : j]}$.
    Each substring is deterministically truncated or padded with trailing zeros to length $n$.
    
    For each such candidate $z$ and each $(\sk_t, r_t, e_t)$ in the secret key, the decoder computes the embedding $\emb(z, r_t)$, unmasks it by computing $\hat{a}_t = \emb(z, r_t) \oplus e_t$, and applies the Hamming PRC decoder with secret key $\sk_t$.
    The decoder outputs $1$ as soon as any trial of $\dec_h$ accepts; if none of the trials accept, it outputs $\bot$.
    
    This procedure exploits the one-sided success guarantee of the CGK embedding: it suffices that one good block and one good seed succeed.
\end{itemize}

\medskip

\begin{breakablealgorithm}
\caption{CGK-Based Zero-Bit Public-Key Edit PRC}
\label{alg: cgk}
\begin{algorithmic}[1]
\Require Security parameter $\lambda$, $0 < \gamma < 1 / 2$, $n = n(\lambda)$, $\ell = n^{\frac{1}{2 \gamma} - 1}$
\Statex \hspace{\algorithmicindent} Zero-bit public-key PRC $\prc_h = (\key_h, \enc_h, \dec_h)$ with codeword length $1.5 n$
\Statex \hspace{\algorithmicindent} CGK embedding function $\emb : \{0, 1\}^n \times \{0, 1\}^{3 n} \to \{0, 1\}^{1.5 n}$
\Statex \hspace{\algorithmicindent} CGK decoding function $\emb^{-1} : \{0, 1\}^{1.5 n} \times \{0, 1\}^{3 n} \to \{0, 1\}^n$

\Function{$\key$}{$1^\lambda$}
    \State Use $\key_h(1^\lambda)$ to sample $\ell$ independent pairs of keys $(\pk_t, \sk_t)_{t = 1}^\ell$
    \State Sample $\ell$ independent random seeds $r_1, \ldots, r_\ell \gets \{0,1\}^{3 n}$
    \State Sample $\ell$ independent random masks $e_1, \ldots, e_\ell \gets \{0,1\}^{1.5 n}$
    \State \Return $(\pk, \sk) = \big( (\pk_t, r_t, e_t)_{t = 1}^\ell, (\sk_t, r_t, e_t)_{t = 1}^\ell \big)$
\EndFunction
\Statex

\Function{$\enc$}{$1^\lambda$, $\pk$, $1$} \Comment{The only message is $1$ since this is a zero-bit PRC}
    \For{$t = 1, 2, \ldots, \ell$}
        \State Define $a_t \gets \enc_h(1^\lambda, \pk_t, 1)$
        \State Define $a'_t = a_t \oplus e_t$ \Comment{Apply one-time pad in Hamming space}
        \State Define $y_t \gets \emb^{-1}(a'_t, r_t)$ 
    \EndFor
    \State \Return $y = y_1 \circ \cdots \circ y_\ell$
\EndFunction
\Statex

\Function{$\dec$}{$1^\lambda$, $\sk$, $y'$} \Comment{Return $\bot$ when $y'$ is far from all valid codewords}
    \For{all intervals $[i : j]$ of $y'$} \Comment{Use a sliding window to hit a good block}
        \State Let $z = y'_{[i : j]}$ \Comment{Try all consecutive substrings of $y'$}
        \State Deterministically truncate $z$ to its first $n$ bits if $|z| > n$ 
        \State Or pad $z$ with trailing zeros to length $n$ otherwise 
        \For{$t = 1, 2, \ldots, \ell$}
            \If{$\dec_h(1^\lambda, \sk_t, \emb(z, r_t) \oplus e_t) \ne \bot$} \Comment{Recover and decode $z$ in $\prc_h$}
                \State \Return $1$ \Comment{Return $1$ if any instance succeeds}
            \EndIf
        \EndFor
    \EndFor
    \State \Return $\bot$ \Comment{Return $\bot$ if no instance succeeds}
\EndFunction
\Statex
\end{algorithmic}
\end{breakablealgorithm}

\subsection{Analysis of \cref{alg: cgk}} \label{sec:zero_bit_PRC_CGK_analysis}

Now we are ready to prove the main theorem. We restate \Cref{thm: CGK_construction} here. We only prove for the public-key case. The secret-key case is similar.

\CGKconstruction*

Let $\prc = (\key,\enc,\dec)$ be the scheme constructed in \Cref{alg: cgk}, and $N = N(\lambda) = n(\lambda) \cdot \ell = n^{\frac{1}{2 \gamma}}$ be the codeword length of $\prc$. Recall that the decoder enumerates all substrings $z = y'_{[i : j]}$, deterministically truncates/pads each $z$ to length $n$, and tries all $(\sk_t, r_t, e_t)_{t = 1}^\ell$. The decoder accepts if any call to $\dec_h(1^\lambda, \sk_t, \emb(z, r_t) \oplus e_t)$ accepts.

\paragraph{Robustness.}
Since it is a zero-bit PRC, the only message is $1$. For each $t \in [\ell]$, let $a_t \gets \enc_h(1^\lambda, \pk_t, 1)$ be the Hamming PRC codeword sampled in the encoding procedure, $a'_t = a_t \oplus e_t$ be the masked block, and let $y_t \gets \emb^{-1}(a'_t, r_t)$ be the corresponding block. Let $y = y_1 \circ \cdots \circ y_\ell$ denote the final codeword computed by $\enc(1^\lambda, \pk, 1)$, and let $y' = \sE(y)$ be the corrupted codeword through a $p(\cdot)$-bounded edit channel $\sE$ with $p : x \mapsto c_2 / x^\gamma$.

We partition $y$ into $\ell$ consecutive blocks $y_1,\dots,y_\ell$, and assign each edit operation in the edit transcript of $(y, y')$ to exactly one block as follows: edits strictly inside $y_t$ are assigned to block $t$, and edits crossing the boundary between $y_t$ and $y_{t+1}$ are assigned to block $t$. This induces a corresponding partition
$$
y' = y'_1 \circ \cdots \circ y'_\ell,
$$
where $y'_t$ is the substring of $y'$ aligned with $y_t$ under the edit transcript.

Since the total edit distance is at most $p(N) \cdot |y|$, by Markov's inequality, the number of
blocks $t$ for which $\ed(y_t, y'_t) > 2 p(N) \cdot |y_t|$ is at most $\ell/2$. Hence, at least $\ell/2$ blocks are \emph{good}, meaning
$$
\ed(y_t, y'_t) \le 2 p(N) \cdot |y_t| = 2 n \cdot \frac{c_2}{(n^{1 / (2 \gamma)})^\gamma} = 2 c_2 \sqrt{n}.
$$
Recall that the decoder enumerates all substrings of $y'$. Therefore for each good block $t$, there exists some iteration in which the candidate substring $z$ equals $y'_t$. Then $z$ will be deterministically truncated/padded to length $n$, which may at most double the edit distance, ensuring $\ed(y_t, z) \le 4 c_2 \sqrt{n}$.

Crucially, we establish the statistical independence between $(y_t, z)$ and the embedding seed $r_t$. Fix the underlying key $(\pk_t, \sk_t)$ and the internal randomness of $\enc_h$, so that $a_t$ is a fixed string. Since $e_t \sim U_{1.5 n}$ is sampled independently, the masked codeword $a'_t = a_t \oplus e_t$ is uniformly distributed over $\{0, 1\}^{1.5 n}$ and strictly independent of $r_t$. By the second property of \Cref{lem: CGK_dec_props}, the output $y_t = \emb^{-1}(a'_t, r_t)$ is uniformly distributed over $\{0, 1\}^n$, and its distribution is independent of $r_t$. Furthermore, since the edit channel $\sE$ is oblivious (its error operations are fixed beforehand and do not depend on the codewords or the keys), the edited string $y'$ and the substring $z$ depend exclusively on $y$ and the predetermined channel $\sE$. Therefore, the random variables $(y_t, z)$ are entirely independent of $r_t$.

Fix a good block $t$. By the first property of \Cref{lem: CGK_dec_props}, with probability at least $1 - \negl(\lambda)$ over the randomness of $r_t$, for any constant $c_3 > 0$,
$$
\ham\big(\emb(y_t, r_t), a'_t\big) \le \left(\frac{3}{8} + c_3\right) n.
$$

Since $(y_t, z)$ is independent of $r_t$, we can apply \Cref{lem: CGK_emb_prop} by conditioning on any fixed values $y_t = v$ and $z = w$ such that $\ed(v, w) \le 4 c_2 \sqrt{n}$. Over the randomness of $r_t$, for any constant $c_4 > 0$,
$$
\Pr_{r_t}\big[ \ham\big(\emb(v, r_t), \emb(w, r_t)\big) \le c_4 \cdot (\ed(v, w))^2 \big] \ge 1 - 6 / \sqrt{c_4}.
$$

Taking the expectation over all such valid pairs $(v, w)$, we have with probability at least $1 - 6 / \sqrt{c_4}$ over $r_t$:
$$
\ham\big(\emb(y_t, r_t), \emb(z, r_t)\big) \le c_4 \cdot (\ed(y_t, z))^2 \le 16 c_4 c_2^2 \cdot n.
$$

Choose $c_4$ such that $1 - 6 / \sqrt{c_4} \ge 1 / 2$, and choose $c_2, c_3$ sufficiently small so that
$$
\frac{3}{8} + 16 c_4 c_2^2 + c_3 \le \frac{3}{8} + 1.5 c_1.
$$
Then we can bound the distance to the original codeword $a_t$:
\begin{align*}
    \ham(\hat{a}_t, a_t) & = \ham\big(\emb(z, r_t) \oplus e_t, a_t\big) \\
      & = \ham\big(\emb(z, r_t), a_t \oplus e_t\big) \\
      & = \ham\big(\emb(z, r_t), a'_t\big) \\
      & \le \ham\big(\emb(z, r_t), \emb(y_t, r_t)\big) + \ham\big(\emb(y_t, r_t), a'_t\big) \\
      & \le 16 c_4 c_2^2 \cdot n + \left(\frac{3}{8} + c_3\right) n \le \left(\frac{1}{4} + c_1\right) \cdot 1.5n.
\end{align*}

Since $\prc_h$ is robust to any $(\frac{1}{4} + c_1)$-bounded Hamming channel and has codeword length $1.5n$, it follows that, with probability at least $1 / 2 - \negl(\lambda)$,
$$
\dec_h\big(1^\lambda, \sk_t, \emb(z, r_t) \oplus e_t\big) = 1.
$$

Finally, note that at least $\ell/2$ blocks are good and each good block succeeds with constant probability over $r_t$, and the decoder tries all $t \in [\ell]$. Additionally, since $(\sk_t, r_t, e_t)_{t = 1}^\ell$ are independent, these good blocks also succeed independently. Hence the probability that none of the trials succeeds is at most $\left(\frac{1}{2} + \negl(\lambda) \right)^{\ell / 2} = \negl(\lambda)$. Thus,
$$
\Pr_{(\pk, \sk) \gets \key(1^\lambda)}[\dec(1^\lambda, \sk, y') = 1] \ge 1 - \negl(\lambda).
$$

\paragraph{Soundness.}
Suppose for contradiction that $\prc$ is not sound. Then there exists a fixed string $y$ such that
$$
\Pr_{(\pk, \sk) \gets \key(1^\lambda)}\big[ \dec(1^\lambda, \sk, y) \neq \bot \big] \ge \frac{1}{\poly(\lambda)}.
$$
By a standard averaging argument, there exist fixed choices of seeds $(r^*_1, \ldots, r^*_\ell)$ and masks $(e^*_1, \ldots, e^*_\ell)$ such that
$$
\Pr_{(\pk_t, \sk_t)_{t=1}^\ell}\big[ \dec(1^\lambda, \sk, y) \neq \bot \big] \ge \frac{1}{\poly(\lambda)},
$$
where the secret key $\sk$ now incorporates these fixed seeds and masks.

By the definition of the decoder, this implies that with non-negligible probability over $(\pk_t, \sk_t)_{t=1}^\ell$, there exist indices $i, j$ and $t$ such that
$$
\dec_h\big(1^\lambda, \sk_t, \emb(z, r^*_t) \oplus e^*_t\big) = 1,
$$
where $z$ is obtained from $y_{[i : j]}$ by a fixed deterministic truncation/padding rule to length $n$.

By a union bound over all $\Theta(n^2)$ choices of $(i, j)$ and $\ell$ blocks, there exist fixed indices $(i^*, j^*, t^*)$ such that
$$
\Pr_{(\pk_{t^*}, \sk_{t^*}) \gets \key_h(1^\lambda)}\big[ \dec_h\big(1^\lambda, \sk_{t^*}, \emb(z^*, r^*_{t^*}) \oplus e^*_{t^*}\big) = 1 \big] \ge \frac{1 / \poly(\lambda)}{\Theta(n^2)  \cdot \ell} = \frac{1}{\poly(\lambda)},
$$
where $z^*$ is the deterministically truncated/padded version of $y_{[i^* : j^*]}$.

Let $y^* = \emb(z^*, r^*_{t^*}) \oplus e^*_{t^*}$. Then we have
$$
\Pr_{(\pk_{t^*}, \sk_{t^*}) \gets \key_h(1^\lambda)}\big[ \dec_h(1^\lambda, \sk_{t^*}, y^*) \neq \bot \big] \ge \frac{1}{\poly(\lambda)},
$$
which contradicts the soundness of $\prc_h$. Therefore, $\prc$ is sound.

\paragraph{Pseudorandomness.}
Since the only message is $1$, we prove that the outputs of $\enc(1^\lambda, \pk, 1)$ are computationally indistinguishable from $U_{n \ell}$ given the public key $\pk$ via hybrid arguments.

We define a sequence of $\ell + 1$ hybrid experiments $H_0, H_1, \ldots, H_\ell$. In hybrid $H_k$, the challenger constructs the public key $\pk$ exactly as in the real scheme, but answers each encoding oracle query for $1$ by generating the $\ell$ blocks $y_1, \ldots, y_\ell$ as follows:
\begin{itemize}
    \item For $t \le k$, the block is generated honestly: $a_t \gets \enc_h(1^\lambda, \pk_t, 1)$ and $y_t \gets \emb^{-1}(a_t \oplus e_t, r_t)$.
    
    \item For $t > k$, the block is generated randomly: $a_t \gets U_{1.5n}$ and $y_t \gets \emb^{-1}(a_t \oplus e_t, r_t)$.
\end{itemize}

Observe the endpoints of these hybrids:
\begin{itemize}
    \item In $H_\ell$, all $\ell$ blocks are generated honestly, which perfectly corresponds to the real encoding oracle $\enc(1^\lambda, \pk, \cdot)$.

    \item In $H_0$, all $\ell$ blocks are generated from fresh uniform strings $a_t \sim U_{1.5n}$. For any $t$, $e_t$ is given in the public key, and so $a_t \oplus e_t$ is strictly distributed as $U_{1.5n}$. By the second property of \Cref{lem: CGK_dec_props}, for the \emph{fixed} $r_t$ in the public key, the output $y_t = \emb^{-1}(a_t \oplus e_t, r_t)$ is uniform. Since $a_1, \ldots, a_\ell$ are independent, $y = y_1 \circ \cdots \circ y_\ell$ in $H_0$ is perfectly uniformly distributed over $\{0,1\}^{n \ell}$. Thus, $H_0$ perfectly corresponds to the uniform oracle $U_{n \ell}$.
\end{itemize}

We claim that for any $k \in [\ell]$, the adjacent hybrids $H_{k-1}$ and $H_k$ are computationally indistinguishable. Suppose there exists a PPT distinguisher $\sA$ that distinguishes $H_{k-1}$ and $H_k$. We can construct a PPT distinguisher $\sB$ to break the public-key pseudorandomness of $\prc_h$.

$\sB$ receives a single target public key $\pk^*$ and has access to an oracle $\mathcal{O}^*$, which is either the real encoder $\enc_h(1^\lambda, \pk^*, \cdot)$ or a uniform oracle returning strings in $U_{1.5n}$. $\sB$ simulates the environment for $\sA$ by embedding $\pk^*$ into the $k$-th position:
\begin{itemize}
    \item For $t \ne k$, $\sB$ runs $\key_h(1^\lambda)$ to independently generate $(\pk_t, \sk_t)$.
    
    \item For $t = k$, $\sB$ sets $\pk_k = \pk^*$.
\end{itemize}

$\sB$ then independently samples $r_1, \ldots, r_\ell$ from $\{0,1\}^{3n}$ and $e_1, \ldots, e_\ell$ from $\{0,1\}^{1.5n}$. It constructs the full public key $\pk = (\pk_t, r_t, e_t)_{t=1}^\ell$ and runs $\sA(1^\lambda, \pk)$.

Whenever $\sA$ makes an encoding oracle query for $1$, $\sB$ constructs the intermediate strings $a_1, \ldots, a_\ell$ as follows:
\begin{itemize}
    \item For $t < k$, $\sB$ computes $a_t \gets \enc_h(1^\lambda, \pk_t, 1)$ using the keys it generated.
    
    \item For $t = k$, $\sB$ queries its oracle to obtain $a_k \gets \mathcal{O}^*(1)$.
    
    \item For $t > k$, $\sB$ directly samples $a_t \gets U_{1.5n}$.
\end{itemize}
For each $t \in [\ell]$, $\sB$ computes $y_t \gets \emb^{-1}(a_t \oplus e_t, r_t)$, and returns the concatenated string $y = y_1 \circ \cdots \circ y_\ell$ to $\sA$. When $\sA$ outputs a bit, $\sB$ outputs the exact same bit.

If the oracle $\mathcal{O}^*$ is the real encoder $\enc_h(1^\lambda, \pk^*, \cdot)$, then $a_k$ is a valid codeword, matching the distribution of $H_k$; if $\mathcal{O}^*$ provides uniformly random strings, then $a_k \sim U_{1.5n}$, matching the distribution of $H_{k-1}$. 

Obviously, $\sB$ is a valid PPT distinguisher. By the public-key pseudorandomness of $\prc_h$, the distinguishing advantage between $\mathcal{O}^* = \enc_h$ and $\mathcal{O}^* = U_{1.5n}$ must be negligible, which implies a contradiction. Consequently, the advantage of $\sA$ in distinguishing $H_{k-1}$ and $H_k$ is at most $\negl(\lambda)$.

Finally, the total distinguishing advantage between $H_0$ and $H_\ell$ is at most $\ell \cdot \negl(\lambda) = \negl(\lambda)$. Therefore, the output distribution of $\prc$ is computationally indistinguishable from uniform given the public key $\pk$.

\begin{remark}
    (1) The uniform mask $e_t$ is crucial for decoupling the codeword from the embedding seed $r_t$. \Cref{lem: CGK_emb_prop} relies on a random walk argument that strictly requires the input strings to be statistically \emph{independent} of $r_t$. Without $e_t$, the block $y_t$ (and thus the edited substring $z$) would be heavily correlated with $r_t$, invalidating the probability bound. In contrast, the first property of \Cref{lem: CGK_dec_props} does \emph{not} require the input $a$ to be independent of $r_t$. Its distance bound is unconditionally dominated by the event $h_t(0) = h_t(1) = 0$, which depends solely on $r_t$, allowing a direct application of Hoeffding's inequality regardless of $a$.

    (2) The reduction in \Cref{thm: CGK_construction} is not limited to the binary alphabet. The CGK embedding applies to strings over any constant-size alphabet, and hence the reduction holds verbatim for PRCs over any constant alphabet $\Sigma$.
\end{remark}
\section{Generic Reduction to Zero-Bit}\label{sec:generic_reduction_to_zero_bit}

In this section, we present a general black-box framework for using a zero-bit PRC to construct a constant-rate PRC over a general alphabet, and hence also over the binary alphabet. The construction works as follows. We use a short master seed $s$ both to generate the payload-side randomness and to protect itself through many sparse marker blocks inserted at random locations. The decoder first recovers $s$ from these marker blocks, and then uses $s$ to decode the payload. Our construction can be seen as an extension of \cite{christPseudorandomErrorCorrectingCodes2024} for Hamming errors, but we make several modifications to adapt it to edit errors.

We first describe the building blocks for our reduction framework, and abstract an object named \emph{seeded payload code} in \Cref{sec:generic_reduction_building_blocks}. Then we give detailed descriptions of our reduction in \Cref{sec:generic_reduction_construction}. Finally, we state our main result \Cref{thm:const-reduction} for this section and prove it in \Cref{sec:generic_reduction_analysis}.

\subsection{Building Blocks}\label{sec:generic_reduction_building_blocks}

To construct a constant-rate PRC against any $p(\cdot)$-bounded edit channel for a function $p : \N \to [0, 1]$, we use the following building blocks.
\begin{itemize}
    \item \textbf{Zero-Bit PRC.}
    Let $\lambda$ be the security parameter. Let
    \[
    \prcz=(\keyz,\encz,\decz)
    \]
    be a zero-bit public-key PRC over alphabet $\Sigma$ with codeword length $\ell$, which is robust to any $\pz(\cdot)$-bounded edit channel for a function $\pz : \N \to [0, 1]$. Moreover, $\prcz$ should satisfy the uniformity property defined in the following. Particularly, the zero-bit public-key PRC given by \Cref{alg: cgk} in \Cref{sec:zero_bit_PRC_CGK_construction} satisfies this uniformity property.
    
    \begin{definition}[Uniformity]\phantomsection\label{def:prc_uniformity}
        We say a PRC with codeword length $\ell$ over alphabet $\Sigma$ has the \emph{uniformity property}, if for any fixed message, over the randomness of the key-generation and encoding procedures, the resulting codeword is uniformly distributed over $\Sigma^\ell$.
    \end{definition}
    
    \item \textbf{Seeded Payload Code.}
     We next introduce the notion of a seeded payload code, which will serve as the payload code in our construction. Informally, such a code uses a short random seed and a (public) masking string to encode the message, while satisfying robustness, uniformity, and pseudorandomness.
     
    \begin{definition}[Seeded payload code]\phantomsection\label{def:seeded_payload_code}
        We use $M$ to denote the size of the message space. Let
        \[
        \eccload:[M]\times\{0,1\}^d\times\{0, 1\}^r\to\Sigma^n
        \]
        be the encoding map for the payload code with rate $R=\frac{\log_{|\Sigma|} M}{n}$, which should be deterministic polynomial-time computable. We say $\eccload$ is a \emph{seeded payload code robust to any $\delta$ fraction of edit errors}, if it satisfies the following properties:
        \begin{itemize}
            \item \textbf{Robustness.} There exists a deterministic polynomial-time computable decoding map
            \[
            \eccload^{-1}:\Sigma^\ast\times\{0,1\}^d\times\{0, 1\}^r\to [M]\cup\{\bot\}
            \]
            such that for every $m\in[M]$, with probability at least $1-\negl(\lambda)$ over the choice of
            $s\in\{0,1\}^d$ and $\shiftpk\in\{0, 1\}^r$, every received word $y'$ satisfying
            \[
            \ed\bigl(y',\eccload(m,s,\shiftpk)\bigr)\le \delta n
            \]
            also satisfies
            \[
            \eccload^{-1}(y',s,\shiftpk)=m.
            \]
    
            \item \textbf{Uniformity.} If $\shiftpk$ is sampled uniformly from $\{0, 1\}^r$, then for every fixed
            $m\in[M]$ and $s\in\{0,1\}^d$, the output $\eccload(m,s,\shiftpk)$ is uniformly distributed over $\Sigma^n$.
    
            \item \textbf{Pseudorandomness.} For every PPT adversary $\sA$,
            \[
            \left|
            \Pr_{\shiftpk \gets U_r}\bigl[\sA^{\eccload(\cdot, U_d, \shiftpk)}(1^\lambda,\shiftpk)=1\bigr]
            -
            \Pr_{\substack{\shiftpk \gets U_r \\ \sU}}\bigl[\sA^{\sU}(1^\lambda,\shiftpk)=1\bigr]
            \right|
            \le \negl(\lambda),
            \]
            where $\eccload(\cdot, U_d, \shiftpk)$ denotes the oracle that first samples $s \gets U_d$ and then returns $\eccload(\cdot, s, \shiftpk)$, and $\sU$ denotes an oracle that returns a fresh uniform sample from $\Sigma^n$ on each query. Note that the distinguisher knows $\shiftpk$ since we will put it in the public key.
        \end{itemize}
    \end{definition}
    
    \item \textbf{Hamming ECC.}
    Let
    \[
    \eccseed:\{0,1\}^d\to\{0,1\}^{d'}
    \]
    be a seed code of constant rate that can be efficiently decoded from a $\pseed$ fraction of Hamming errors, where $0 < \pseed < 1 / 4$ is a constant. Such codes exist for all sufficiently large $d$. For example, one may use a concatenated code.

    We take $d=\Theta(n^{\gamma})$ for some constant $0 < \gamma < 1 / 4$. Since $\eccseed$ has constant rate, it follows that $d'=O(d)$.
\end{itemize}

\subsection{Detailed Construction}\label{sec:generic_reduction_construction}

Now we describe the construction of our PRC.
\begin{itemize}
    \item \textbf{Key Generation:}
    For each $j\in[2 d']$, sample an independent key pair
    \[
    (\pk_j,\sk_j)\gets \keyz(1^\lambda).
    \]
    Also sample a uniformly random string $\shiftpk \gets U_r$.
    Set
    \[
    \pk:=(\pk_1,\dots,\pk_{2 d'},\shiftpk),
    \qquad
    \sk:=(\sk_1,\dots,\sk_{2 d'},\shiftpk).
    \]

    \item \textbf{Encoding:}
    Given a message $m\in[M]$, the encoder first samples a short seed
    \[
    s\gets\{0,1\}^d.
    \]
    It then encodes the seed using the seed code and appends a padding of $d'$ ones:
    \[
    s':=\eccseed(s) \circ 1^{d'}\in\{0,1\}^{2 d'}.
    \]
    For each $i\in[2 d']$, the encoder uses the public key $\pk_i$ to encode the bit $s'_i$, obtaining a marker block
    \[
    u_i \gets
    \begin{cases}
    \encz(1^\lambda,\pk_i,1), & \text{if } s'_i=1,\\
    U_\ell, & \text{if } s'_i=0.
    \end{cases}
    \]
    Next, the encoder computes the payload codeword
    \[
    c:=\eccload(m,s,\shiftpk)\in\Sigma^n.
    \]
    Finally, the encoder samples insertion locations $\loc_1, \ldots, \loc_{2 d'}$ and interleaves the marker blocks $u_1,\dots,u_{2 d'}$ into the payload codeword $c$. More precisely, we view the final codeword as occupying coordinates in $\left[n+2 d'\ell\right]$.
    The encoder samples $2 d'$ pairwise disjoint intervals
    \[
    I_1 = [\loc_1, \loc_1 + \ell - 1], \dots,I_{2 d'} = [\loc_{2 d'}, \loc_{2 d'} + \ell - 1]\subseteq \left[n+2 d'\ell\right],
    \]
    each of length $\ell$, uniformly at random subject to the condition that they are pairwise disjoint. Equivalently, one may sample candidate intervals independently and resample until all chosen intervals are pairwise disjoint. For each $i\in[2 d']$, the marker block $u_i$ is placed on the interval $I_i$, while the symbols of the payload codeword $c$ are placed in the remaining coordinates in their natural order. The resulting word is denoted by $y$.

    \item \textbf{Decoding:}
    Given a received word $y'$, the decoder first searches for candidate marker blocks. More precisely, for each index $i\in[2 d']$, it applies the zero-bit decoder $\decz(1^\lambda,\sk_i,\cdot)$ to candidate substrings of $y'$, and uses the resulting outputs to reconstruct a corrupted version $\tilde{s}$ of the protected seed string $s'$. It then checks whether the padding portion of $\tilde{s}$ is consistent, and decodes the seed-encoding portion using the decoder of $\eccseed$ to recover the seed $s$. If either step fails, the decoder outputs $\bot$.
    
    Once we successfully recover $s$, we can run the decoder $\eccload^{-1}$ for the payload code using $s$ and $\shiftpk$, and output the recovered message.
\end{itemize}

\begin{breakablealgorithm}
\caption{Generic Construction from a Zero-Bit Public-Key Edit PRC}\label{alg:const-rate}
\begin{algorithmic}[1]
\Require Security parameter $\lambda$, $n = n(\lambda)$, $d = d(\lambda)$, $\ell = \ell(\lambda)$, message space $[M]$
\Statex \hspace{\algorithmicindent}Zero-bit public-key PRC $\prcz=(\keyz,\encz,\decz)$ over alphabet $\Sigma$
\Statex \hspace{\algorithmicindent}Seeded payload code $\eccload : [M] \times \{0, 1\}^d \times \{0, 1\}^r \to \Sigma^n$
\Statex \hspace{\algorithmicindent}Seed code $\eccseed : \{0, 1\}^d \to \{0, 1\}^{d'}, d' = O(d)$

\Function{$\key$}{$1^\lambda$}
    \For{$j=1$ to $2 d'$}
        \State $(\pk_j,\sk_j)\gets \keyz(1^\lambda)$
    \EndFor
    \State Sample a uniformly random shifting vector $\shiftpk \gets U_r$
    \State $\pk = (\pk_1,\dots,\pk_{2 d'},\shiftpk)$
    \State $\sk = (\sk_1,\dots,\sk_{2 d'},\shiftpk)$
    \State \Return $(\pk,\sk)$
\EndFunction

\Statex

\Function{$\enc$}{$1^\lambda,\pk,m$}
    \State Sample $s\gets\{0,1\}^d$
    \State Compute $s'\gets \eccseed(s) \circ 1^{d'}$
    \For{$i=1$ to $2 d'$}
        \If{$s'_i=1$}
            \State $u_i\gets \encz(1^\lambda,\pk_i,1)$ \Comment{For $s_i' = 1$, the marker block is a PRC codeword}
        \Else
            \State $u_i\gets U_\ell$ \Comment{For $s_i' = 0$, the marker block is a random string}
        \EndIf
    \EndFor
    \State Compute the payload codeword $c = \eccload(m,s,\shiftpk)$
    \State Interleave the marker blocks $u_1,\dots,u_{2 d'}$ into the payload codeword $c$ as described
    \State \Return the resulting word $y$
\EndFunction

\Statex

\Function{$\dec$}{$1^\lambda,\sk,y'$}
    \State Initialize $\tilde{s} = \bot^{2 d'}$ \Comment{Step 1: Recover the seed $s$}
    \For{$i=1$ to $2 d'$}
        \For{all intervals $[l : r]$ of $y'$}
            \State Run $\decz(1^\lambda, \sk_i,y'_{[l : r]})$ \Comment{Try all substrings of $y'$} \phantomsection\label{line:red:enum}
        \EndFor
        \State Set $\tilde{s}_i= 1$ if the decoder succeeds in any trial; otherwise set $\tilde{s}_i= 0$
    \EndFor
    \If{$\tilde{s}_{[d' + 1, 2 d']} = 0^{d'}$} \Comment{Check consistency of the padding portion } \phantomsection\label{line:red:check}
        \State \Return $\bot$
    \EndIf
    \State Decode $\tilde{s}_{[1 : d']}$ using the seed decoder $\eccseed^{-1}$ to recover $s$ \phantomsection\label{line:red:rec}
    \If{seed decoding fails}
        \State \Return $\bot$
    \EndIf
    \State \Return $\eccload^{-1}(y',s,\shiftpk)$ \Comment{Step 2: Use the payload decoder to recover the message}
\EndFunction
\end{algorithmic}
\end{breakablealgorithm}

\subsection{Analysis of the \Cref{alg:const-rate}}\label{sec:generic_reduction_analysis}

\begin{theorem}\label{thm:const-reduction}
    Let $\lambda$ be a security parameter, and:
    \begin{itemize}
        \item $\prcz = (\keyz, \encz, \decz)$ is a zero-bit public-key (secret-key) PRC with codeword length $\ell$ that is robust to any $\pz(\cdot)$-bounded edit channel for the function $\pz : \N \to [0, 1]$. Moreover, $\prcz$ satisfies the uniformity property (see \Cref{def:prc_uniformity}).

        \item $\eccload : [M] \times \{0, 1\}^d \times \{0, 1\}^r \to \Sigma^n$ is a seeded payload code (see \Cref{def:seeded_payload_code}) with rate $R = \log_{|\Sigma|} M / n$ that is robust to any $\delta$ fraction of edit errors.
        
        \item $\eccseed : \{0, 1\}^d \to \{0, 1\}^{d'}$ is an error-correcting code against any $\pseed$ fraction of Hamming errors, where $0 < \pseed < 1 / 4$ is a constant.
    \end{itemize}
    
    Let $N=n+2 d'\ell$ be the total codeword length. Suppose the parameters satisfy $d'=\ell=\Theta(n^{\gamma})$ for some constant $0 < \gamma < 1 / 4$. Moreover, let $p : \N \to [0, 1]$ be a function such that $p(N) \le (1 - \varepsilon) \cdot \min\{\delta, \pz(\ell) \cdot \pseed\}$ for some constant $\varepsilon>0$. Then the scheme $\prc=(\key,\enc,\dec)$ constructed in \Cref{alg:const-rate} is a public-key (secret-key) PRC that is robust against any $p(\cdot)$-bounded edit channels $\sE$, with rate
    \[
    \frac{\log_{|\Sigma|} M}{n+2d'\ell}=R-o(1).
    \]
\end{theorem}

\paragraph{Robustness.}
We now establish the robustness of the edit PRC. Observe that the zero-bit PRC component and the payload codeword are independent, and each is uniformly distributed over the randomness of key generation and encoding. Combining with the uniformity property of the payload code $\eccload$, the final codeword is also uniformly distributed.

Let $y$ denote the transmitted word, and let $y'=\sE(y)$ be the received word. By the uniformity property, the edit pattern imposed by the channel is independent of the random placement of the blocks $u_i$. Therefore, for the purpose of analysis, we may equivalently imagine that the channel first fixes its edit pattern, and only afterwards the encoder samples the random locations of the blocks $u_i$.

We charge each edit operation in the edit transcript of $(y, y')$ to exactly one position as follows: for any $j \in [N]$, edit operations between two positions $y_j$ and $y_{j+1}$ are charged to position $j$. For each position $j\in[N]$, let $w_j$ denote the number of edit operations charged to that position. If the marker block $u_i$ is placed on the interval $I_i=[\loc_i,\loc_i+\ell-1]\subseteq[N]$, then the total number of edit operations charged to the interval occupied by $u_i$ is $c_i:=\sum_{j\in I_i} w_j$.

\begin{claim}\label{claim:bad-int}
    Call an interval $I_i$ \emph{bad} if $c_i > \pz(\ell) \cdot \ell$. Then, with probability at least $1-\negl(\lambda)$ over the random choice of the insertion locations, the numbers of bad intervals with $i \in [1, d']$ and $i \in [d' + 1, 2 d']$ are both at most $\pseed \cdot d'$.
\end{claim}

\begin{proof}
We first analyze an idealized model in which the interval locations are sampled independently. We will remove this assumption at the end.

Recall that the total number of edit operations is at most $p(N) \cdot N$. Fix a block $u_i$. The interval $I_i$ has length $\ell$, and its starting location is sampled uniformly from at least $N-\ell+1$ possible positions. Since each charged error can affect at most $\ell$ candidate intervals, for any $i \in [2 d']$, we have
\[
\E[c_i]\le \frac{p(N) \cdot N \ell}{N-\ell+1} \le (1 + o(1)) \cdot p(N) \cdot \ell
\]
Then by Markov's inequality, for any $i \in [2 d']$,
\[
\Pr[c_i > \pz(\ell) \cdot \ell] \le \frac{\E[c_i]}{\pz(\ell) \cdot \ell} \le (1 - \varepsilon + o(1)) \cdot \pseed.
\]

Let $X_1$ and $X_2$ denote the numbers of bad intervals with $i \in [1, d']$ and $i \in [d' + 1, 2 d']$, respectively. Under the independence assumption, the indicators of the events ``$c_i > \pz(\ell) \cdot \ell$'' are independent, and hence by Hoeffding's inequality,
\[
\text{both }\Pr\left[X_1 > \pseed \cdot d' \right]\text{ and }\Pr\left[X_2 > \pseed \cdot d'\right] \le 2 \exp\big(-(2 - o(1)) \cdot \varepsilon^2 \pseed^2 d'\big) = \negl(\lambda).
\]
Hence, under the independence assumption, with probability at least $1-\negl(\lambda)$, the numbers of bad intervals with $i \in [1, d']$ and $i \in [d' + 1, 2 d']$ are both at most $\pseed \cdot d'$.

Finally, we account for the non-overlapping constraint. One sampling trial fails if two sampled intervals intersect. The probability of this event is
\[
q
=
1-\left(1-O\left(\frac{\ell d'}{N}\right)\right)^{2 d'}
=
O\left(\frac{\ell d'^2}{N}\right).
\]
The encoder therefore repeats the sampling procedure independently up to $\lambda$ times, stopping once it finds a collection of pairwise disjoint intervals. The probability that all $\lambda$ trials fail is at most $q^\lambda=\negl(\lambda)$.
\end{proof}

We next show that every non-bad block $u_i$ can be recovered correctly with probability at least $1-\negl(\lambda)$. Consequently, by a union bound over all non-bad blocks, with probability at least $1-\negl(\lambda)$, every non-bad block is recovered correctly.
\begin{itemize}
    \item \textbf{Case $s'_i=1$.}
    In this case, the decoder enumerates all substrings $y'_{[l:r]}$ in \alglineref{line:red:enum}. Since $I_i$ is not bad, the total number of edit operations charged to $I_i$ is at most $\pz(\ell) \cdot \ell$. Therefore, among the candidate substrings examined by the decoder, there exists one that contains the corrupted image of $u_i$ and differs from the original block by at most $\pz(\ell) \cdot \ell$ edits. By the robustness of $\prcz$, the decoder successfully recovers this block, and hence outputs $\tilde{s}_i=1$ with probability at least $1-\negl(\lambda)$.

    \item \textbf{Case $s'_i=0$.}
    In this case, the encoder places an independent uniform string $u_i\leftarrow U_\ell$ on the interval $I_i$. For every candidate substring considered by the decoder, this string remains independent of the key pair $(\pk_i,\sk_i)$. Therefore, by the soundness of $\prcz$, the probability that any such candidate substring is incorrectly accepted is at most $\negl(\lambda)$. Taking a union bound over all candidate substrings examined for the $i$-th position, we conclude that $\tilde{s}_i=0$ with probability at least $1-\negl(\lambda)$.
\end{itemize}

Conditioned on the event in \cref{claim:bad-int} and on the event that all non-bad blocks are recovered correctly, we know that in both the seed-encoding portion and the padding portion, at most $\pseed \cdot d'$ bits are decoded incorrectly. Hence, the padding check in \alglineref{line:red:check} does not output $\bot$. And by the decoding guarantee of $\eccseed$, the seed $s$ is correctly recovered in \alglineref{line:red:rec}. In all, we conclude that the decoder recovers the correct seed $s$ with probability at least $1-\negl(\lambda)$.

Finally, conditioned on recovering the correct seed $s$, we view the inserted $\prc_0$ blocks as additional insertion errors relative to the payload codeword. The underlying channel contributes at most $p(N)$ fraction of edit errors, while the total length of all inserted $\prc_0$ blocks is $2d'\ell = \Theta(N^{2\gamma})$. Hence the payload codeword experiences at most $p(N)N + \Theta(N^{2\gamma}) \leq \delta N$ edit errors in total since $p(N)\leq (1-\varepsilon) \delta$. The robustness of the payload code implies that, for every message $m$, the decoder outputs the correct message with probability at least $1-\negl(\lambda)$.

\paragraph{Soundness.}
We now establish the soundness of the resulting PRC.

Fix any received word $y'\in\Sigma^\ast$. We analyze the seed-recovery stage of the decoder. For each index $i\in[2 d']$ and each candidate interval $I$, the decoder applies $\decz(1^\lambda,\sk_i,\cdot)$ to the corresponding substring of $y'$. By the soundness of $\prcz$, for every fixed pair $(i, I)$,
\[
\Pr_{(\pk_i, \sk_i)\gets \keyz(1^\lambda)}\bigl[\decz(1^\lambda,\sk_i,y'_I)\neq \bot\bigr]\le \negl(\lambda).
\]
Thus, except with negligible probability, the decoder rejects on that candidate interval. Since there are $2 d'$ choices of $i$ and at most polynomially many candidate intervals $I$, by a union bound,
\[
\Pr_{(\pk, \sk)\gets \key(1^\lambda)}\bigl[\text{Some invocation of }\decz\text{ accepts}\bigr]
\le 2 d' \cdot \poly(n) \cdot \negl(\lambda) = \negl(\lambda).
\]

Hence, with probability at least $1-\negl(\lambda)$, every invocation of $\decz$ rejects, which implies that the recovered string $\tilde{s}$ is the all-zero string. Therefore the padding check in the decoding algorithm fails, and the decoder outputs $\bot$.

\paragraph{Pseudorandomness.}
We show that the pseudorandomness of the final construction follows from the pseudorandomness of the underlying zero-bit PRC together with the pseudorandomness of the payload code. Informally, the proof proceeds in two steps. First, we replace the marker blocks produced by the zero-bit PRC with independent uniform blocks; this changes the distribution only negligibly by the security of $\prcz$. Second, conditioned on the marker blocks already being uniform, the payload part is pseudorandom by assumption, and hence the entire interleaved codeword is pseudorandom.

Formally, we define the following sequence of hybrid experiments:

\begin{itemize}
    \item \textbf{Hybrid $H_0$ (Real World).} The challenger generates $(\pk, \sk) \leftarrow \key(1^\lambda)$ honestly. When $\sA$ makes an encoding query for message $m$, the challenger computes $c \leftarrow \enc(1^\lambda, \pk, m)$ using the real encoding algorithm and returns $c$.
    
    \item \textbf{Hybrid $H_1$ (Uniform Marker Blocks).} Same as $H_0$, except that for each encoding query, the marker blocks $u_1, \ldots, u_{2d'}$ are sampled as independent uniform strings from $\Sigma^\ell$ instead of being generated via $\encz$. Specifically, we replace $u_i \leftarrow \encz(1^\lambda, \pk_i, 1)$ when $s'_i = 1$ with $u_i \leftarrow U_\ell$ for all $i \in [2d']$.
    
    \item \textbf{Hybrid $H_2$ (Uniform Payload).} Same as $H_1$, except that the payload codeword $c = \eccload(m, s, \shiftpk)$ is replaced with a uniformly random string $c \leftarrow U_n$.
    
    \item \textbf{Hybrid $H_3$ (Ideal World).} Same as $H_2$, except that the entire output is a uniformly random string from $\Sigma^N$. This corresponds to the oracle $U_N$ in the theorem statement.
\end{itemize}

We now bound the distinguishing advantage between adjacent hybrids.

\begin{claim}
    $H_0$ and $H_1$ are computational indistinguishable even with $\pk$.
\end{claim}

\begin{proof}
    We use a hybrid argument over the $2d'$ marker blocks. Let $H_{0,j}$ be an experiment where the first $j$ marker blocks are uniform, and the remaining $2d'-j$ blocks are generated honestly. Note that $H_{0,0} = H_0$ and $H_{0,2d'} = H_1$. For any $j \in [2 d']$, suppose $\sA$ distinguishes $H_{0, j - 1}$ and $H_{0, j}$ with non-negligible advantage. We construct an adversary $\sB$ against $\prcz$:
    \begin{enumerate}
        \item $\sB$ receives a target public key $\pk^*$ and an oracle $\sO^*$ which is either $\encz(1^\lambda, \pk^*, 1)$ or $U_\ell$.
        
        \item $\sB$ samples $(\pk_i, \sk_i) \leftarrow \keyz(1^\lambda)$ for $i \neq j$ and $\shiftpk \leftarrow U_r$.
        
        \item Run $\sA\big(1^\lambda, \pk = (\pk_1, \dots, \pk^*, \dots, \pk_{2d'}, \shiftpk)\big)$. When $\sA$ queries $m$:
        \begin{itemize}
            \item $\sB$ samples $s \leftarrow \{0,1\}^d$ and computes $s' = \eccseed(s) \circ 1^{d'}$.
            
            \item For $i < j$, $\sB$ sets $u_i \leftarrow U_\ell$.
            
            \item For $i = j$: If $s'_j = 1$, $\sB$ sets $u_j \leftarrow \sO^*$; otherwise, $u_j \leftarrow U_\ell$.
            
            \item For $i > j$: If $s'_i = 1$, $\sB$ sets $u_i \leftarrow \encz(1^\lambda, \pk_i, 1)$; otherwise, $u_i \leftarrow U_\ell$.
            
            \item $\sB$ computes $c = \eccload(m, s, \shiftpk)$ and returns the interleaved result to $\sA$.
        \end{itemize}
    \end{enumerate}

    If $\sO^*$ is the real encoder, $\sB$ perfectly simulates $H_{0,j-1}$. If $\sO^*$ is the uniform oracle, $\sB$ simulates $H_{0,j}$. The two hybrids differ only in the distribution of the $j$-th marker block when $s'_j = 1$: in $H_{0,j-1}$ it is generated as $\encz(1^\lambda,\pk^*,1)$, whereas in $H_{0,j}$ it is sampled uniformly from $U_{\ell}$. By the pseudorandomness of $\prcz$, these two cases are computationally indistinguishable given $\pk^{\ast}$, and therefore the distinguishing advantage of $\sA$ between $H_{0,j-1}$ and $H_{0,j}$ is at most $\negl(\lambda)$. Summing over $2d'$ hybrid steps, the total distinguishing advantage between $H_0$ and $H_1$ is still $\negl(\lambda)$. 
\end{proof}

\begin{claim}
    $H_1$ and $H_2$ are computational indistinguishable even with $\pk$.
\end{claim}

\begin{proof}
    In $H_1$, the marker blocks $u_i$ are all independent uniform strings. This means they are statistically independent of the seed $s$. Therefore, the adversary's view of $s$ through the marker blocks is null. We can reduce this directly to the pseudorandomness of $\eccload$:
    \begin{enumerate}
        \item $\sB$ receives $\shiftpk$ and an oracle $\sO$ which is either $\eccload(\cdot, U_d, \shiftpk)$ or $U_n$.
        
        \item $\sB$ samples all $(\pk_i, \sk_i) \gets \keyz(1^\lambda)$ honestly and runs $\sA\big(1^\lambda, \pk = (\pk_1, \ldots, \pk_{2 d'}, \shiftpk)\big)$.
        
        \item For each query $m$ made by $\sA$, $\sB$ sets all $u_i \leftarrow U_\ell$, obtains $c \leftarrow \sO(m)$, and returns the interleaved result to $\sA$.
    \end{enumerate}
    If $\sO$ is the real payload oracle, this is exactly $H_1$. If it is uniform, this is $H_2$. By the pseudorandomness of the payload code $\eccload$, the advantage is $\negl(\lambda)$.
\end{proof}

Finally, note that in $H_2$, the payload and all marker blocks are independent uniform strings. Their interleaving is identically distributed to a uniform string of length $N$. Thus the advantage between $H_2$ and $H_3$ is $0$. In all, the advantage between $H_0$ and $H_3$ is $\negl(\lambda)$, which completes the proof for pseudorandomness.

\bigskip
\medskip

At the end of this section, we state a direct corollary of our reduction. We may choose a binary code of constant rate with efficient decoding radius $\frac{1}{4}-\varepsilon$ as $\ecc_{\seed}$. This gives:

\begin{corollary}\label{cor:const-reduction-fix-seed-code}
    Let $\lambda$ be the security parameter, and let $\prcz$ and $\eccload$ be as in \Cref{thm:const-reduction}. Suppose the resulting code over $\Sigma$ has codeword length $N$. Let $p:\mathbb{N}\to[0,1]$ be such that
    \[
    p(N)\le (1-\varepsilon)\cdot \min\left\{\delta,\frac{1}{4}\pz(\ell)\right\}
    \]
    for some constant $\varepsilon>0$. Then the scheme $\prc=(\key,\enc,\dec)$ constructed in \Cref{alg:const-rate} is a public-key (respectively, secret-key) PRC robust against every $p(\cdot)$-bounded edit channel $\sE$, and it has rate $R-o(1)$.
\end{corollary}

\section{Constant-Alphabet PRCs with Rate Close to 1}\label{sec:rate-one}

In this section, we first present a construction of a \emph{seeded payload code} (see \cref{def:seeded_payload_code}) over a constant-size alphabet and with a constant decoding radius. We observe that synchronization strings \cite{haeuplerSynchronizationStringsCodes2021} are well suited to our application: they handle edit errors while ensuring the codeword appears random, as required by the definition of a seeded payload code. We use synchronization strings via the standard method \cite{haeuplerSynchronizationStringsCodes2021}: attaching a constant-alphabet random string (with certain synchronizing properties) to the message, extending its alphabet by a constant factor.

To make the codeword uniform over the randomness of $(\sk,\pk)$, we apply the public shifting vector $\shiftpk$ to the codeword at the end of generation. However, as $\shiftpk$ is fixed in the public key, to ensure the codeword remains pseudorandom even when queried multiple times, we must apply additional masking generated by PRGs. This is easily handled as the framework in \cref{sec:generic_reduction_to_zero_bit} allows the encoder and decoder to share common randomness through the shared random seed.

For ease of exposition, throughout this section we assume that the alphabets (e.g. $\Sigma, \sigmsg$) we use are finite fields. This allows us to use addition and subtraction to apply and remove the shifts and masks.

\begin{restatable}[Seeded payload code via synchronization string]{theorem}{thmconstrateprc}\label{thm:const_rate_prc}
	Let $\eccmsg: [M] \to \sigmsg^{n}$ be an ECC robust to any $\pmsg$ fraction of half-errors.

	Then \cref{alg:const_rate_prc} yields a seeded payload code $\eccload:[M]\times\{0,1\}^d\times\{0, 1\}^r\to\Sigma^{n}$
	that is robust to any $\pload$ fraction of edit errors,
	provided that the block length $n= \Omega(\lambda)$,
	$\pload\le \pmsg-\delta$ for some constant $\delta>0$, and the alphabet size $|\Sigma|\ge|\sigmsg|\cdot \Omega(\delta^{-4})$.
\end{restatable}
We use the standard \emph{half-error} metric: each erasure
counts as $1$, each substitution as $2$, so the total half-error count
is $2e + s$.

We will first describe some background on synchronizing strings in \cref{sec:sync_string}, then present the detailed construction, and prove \cref{thm:const_rate_prc} in \cref{sec:analysis_const_rate_prc}.

Finally in \cref{sec:LAR_put_together}, we combine our construction of the seeded payload code with appropriate error-correcting codes to get a multi-bit to zero-bit reduction in \cref{cor:multi_zero_redu_sync}, following \cref{thm:const-reduction}.

We further plug in different zero-bit PRC constructions to obtain a multi-bit edit PRC with a constant rate, these results are stated in \cref{cor:body-rateone-result1} and \cref{cor:body-rateone-result2}.

\subsection{Synchronization Strings and Self-Matching}
\label{sec:sync_string}
Synchronization strings, introduced by Haeupler and Shahrasbi~\cite{haeuplerSynchronizationStringsCodes2021}, provide a way to recover approximate alignment under insertions and deletions. In our proof, we only use a relaxed property of synchronization strings, namely the \emph{self-matching} property, together with the corresponding decoding guarantee from~\cite[Section~6.3]{haeuplerSynchronizationStringsCodes2021}. We recall below the notions that are needed later.

\begin{definition}[Monotone matching {\cite[Definition~6.1]{haeuplerSynchronizationStringsCodes2021}}]
	Let $S, S' \in \Sigma^*$. A \emph{monotone matching} between $S$ and $S'$ is a set of pairs $M=\{(a_1,b_1),\cdots,(a_m,b_m)\}$
	such that
	\[
		a_1<\cdots<a_m,\quad b_1<\cdots<b_m,\quad S_{a_i}=S'_{b_i} \quad \text{for all } i\in[m].
	\]
	Equivalently, a monotone matching is a common subsequence together with the indices realizing it.
\end{definition}

\begin{definition}[$\varepsilon$-self-matching]\label{def:self_matching}
	A string $S\in\Sigma^n$ is said to be \emph{$\varepsilon$-self-matching} if for every monotone matching $M=\{(a_1,b_1),\cdots,(a_m,b_m)\}$ between $S$ and itself, the number of pairs with $a_i\neq b_i$ is at most $\varepsilon n$.
\end{definition}

Intuitively, the $\varepsilon$-self-matching property says that one cannot find a large monotone matching between two copies of the same string that mismatches many positions. This is exactly the global non-self-similarity condition needed by the alignment algorithm \cite[Algorithm~3]{haeuplerSynchronizationStringsCodes2021}.

For our application, we do not explicitly construct synchronization strings. Instead, after recovering the seed $s$, the decoder reconstructs
\[
	\sigma=\prgsyn(s_1)+\shiftsyn \in \sigsyn^{n},
\]
and we use the fact that $\sigma$ is uniformly random over $\sigsyn^{n}$ because $\shiftsyn$ is uniform. Hence, it suffices to know that a random string is $\varepsilon$-self-matching with high probability.

\begin{theorem}[Random strings are self-matching {\cite[Theorem~6.9]{haeuplerSynchronizationStringsCodes2021}}]
	\label{thm:random_self_matching}
	Let $S$ be a uniformly random string in $\Sigma^n$. Then for every $\varepsilon \in (0,1)$,
	\[
		\Pr[\text{$S$ is $\varepsilon$-self-matching}] \ge 1-\left(\frac{e^2}{\varepsilon^2|\Sigma|}\right)^{n\varepsilon}.
	\]
\end{theorem}

We will also use the decoding guarantee from~\cite{haeuplerSynchronizationStringsCodes2021}. Informally, if a received string differs from an $\varepsilon$-self-matching string by at most $t$ insertions and deletions, then the synchronization algorithm outputs an indexing with at most $t+O(\sqrt{\varepsilon})n$ misdecodings.

\begin{lemma}[Alignment via self-matching strings (implicit in \cite{haeuplerSynchronizationStringsCodes2021})]
	\label{lem:alignment}
	Let $\sigma \in \Sigma_{\mathsf{syn}}^n$ be an $\varepsilon$-self-matching string, and let $\tilde\sigma$ be any string with $\Delta_E(\sigma,\tilde\sigma) \le t$.
	Let $\tilde x \in \sigmsg^{|\tilde\sigma|}$ be produced by applying the same edit operations to a string $x \in \sigmsg^n$.

	There is a polynomial-time procedure that outputs
	$\hat x \in (\Sigma_{\mathsf{msg}} \cup \{\bot\})^n$ such that, letting
	$e$ and $s$ be the numbers of incorrect symbols and erasures in $\hat x$
	relative to the original $x$, we have
	\[
		2e + s \le t + 6n\sqrt{\varepsilon}.
	\]
\end{lemma}

\begin{proof}[Proof sketch]
	The repositioning algorithm of~\cite[Algorithm~3]{haeuplerSynchronizationStringsCodes2021}
	produces a mapping $P : [|\tilde\sigma|] \to [n] \cup \{\bot\}$ from received
	positions to source positions.  Inverting this mapping and applying standard
	conflict resolution (as in the rearrangement procedure
	of~\cite[Algorithm~2]{haeuplerSynchronizationStringsCodes2021}): keeping
	only the unique claimant for each source and marking the rest as
	erasures, gives a candidate received position $r_j$ for each source
	position $j \in [n]$.  We set $\hat x_j = \tilde x_{r_j}$ when
	$r_j \neq \bot$, and $\hat x_j = \bot$ otherwise.

	By~\cite[Theorem~6.13]{haeuplerSynchronizationStringsCodes2021},
	at most $3n\sqrt{\varepsilon}$ successfully transmitted positions are
	misdecoded, each contributing at most $2$ half-errors in the worst case.
	Deletions create at most $d_r \text{ (number of deletions)} \le t$ erasures (each $1$ half-error).
	Summing,
	\[
		2e + s \le d_r + 2 \cdot 3n\sqrt{\varepsilon} \le t + 6n\sqrt{\varepsilon}.
	\]
\end{proof}

\subsection{Construction of the Seeded Payload Code}

We list the components required in our construction:

\begin{itemize}
	\item \textbf{Hamming ECCs.}
	      In our construction, we employ a Hamming error-correcting code $\eccmsg: [M] \to \sigmsg^{\,n}$ that is robust against any $\pmsg$ fraction of \emph{half-errors}, i.e., it admits efficient unique decoding from any combination of a fraction $e$ of substitution errors and a fraction $s$ of erasures as long as $2e + s \leq \pmsg$. We denote its encoding and decoding functions by $\eccmsg(\cdot)$ and $\eccmsg^{-1}(\cdot)$, respectively.

	      An explicit family of such codes with rate arbitrarily close to $1 - \pmsg$ and alphabet size $\poly(1/\varepsilon)$ is given by the AG code; we restate the required guarantee in \cref{lem:ag-code}.

	      The role of $\eccmsg$ is to correct the remaining Hamming-type corruptions that remain after the edit errors have been converted into substitution errors and erasures via the synchronization string.
	\item \textbf{Pseudorandom Generator.}
	      We require a cryptographic pseudorandom generator $\prg$ such that no PPT adversary can distinguish its output from a uniform string with non-negligible advantage. In our scheme, we instantiate the PRG twice to produce masking strings over different alphabets:
	      \begin{itemize}
		      \item $\prgmsg: \{0,1\}^{d_2} \to \sigmsg^n$ is used to mask the encoded payload message, ensuring its pseudorandomness.
		      \item $\prgsyn: \{0,1\}^{d_1} \to \sigsyn^n$ is used to mask the synchronization string, making it computationally indistinguishable from uniform.
	      \end{itemize}
	      These two generators ensure that each component of the assembled codeword satisfies the pseudorandomness requirement.
\end{itemize}
\paragraph{Notation convention.} In this section, we follow a consistent naming convention to distinguish between the different stages of the codeword.
We use a tilde to denote corrupted sequences (e.g., $\tilde{y}$ and $\tilde{\sigma}$). We also use $\tilde{n}$ to denote the length of the corrupted message. For sequences reconstructed during the decoding process, we use a hat (e.g., $\hat{x}$ as the estimate of $x$ and $\hat{m}$ as the recovered message).

\begin{description}[style=nextline]
	\item[Encoding $\eccload(m,s,\shiftpk)$:]
	      \textbf{Applying Synchronization String and Masking.} Parse $\shiftpk$ into $(\shiftmsg,\shiftsyn)$ such that $(\shiftpk)_i=((\shiftmsg)_i,(\shiftsyn)_i)\in \sigmsg\times\sigsyn$. And split $s$ into $(s_1,s_2)$.

	      We then generate the synchronization strings by applying the public shift $\shiftsyn$ to a pseudo-random string $\prgsyn(s_1)$, i.e.,
	      $$\sigma = \prgsyn(s_1)+\shiftsyn\in \sigsyn^{n}.$$
	      For the message part, we compute $\eccmsg(m)\in \sigmsg^{n}$, mask it with $\prgmsg(s_2)$, and apply $\shiftmsg$:
	      $$
		      x = \eccmsg(m) + \prgmsg(s_2) + \shiftmsg \in \sigmsg^{n}.
	      $$
	      We then assemble the edit PRC by combining $x$ and $\sigma$; specifically, we define the assembled sequence $y$ where each entry $y_i \coloneqq (x_i, \sigma_i) \in \sigmsg \times \sigsyn$ (combining $\sigma$ and $x$ symbol-by-symbol lets the decoder recover alignment under insertions and deletions by matching $\tilde\sigma$ against $\sigma$ (see \cref{lem:alignment})).

	\item[Decoding $\eccload^{-1}(\tilde{y},s, \shiftpk)$:]
	      \textbf{Recover the Synchronizing String.}
	      We split $s$ into $s_1$ and $s_2$, and regenerate the synchronization string $\sigma = \prgsyn(s_1) + \shiftsyn$.

	      \textbf{Synchronize.}
	      We denote the synchronization string part of $ \tilde{y}$ as $\tilde{\sigma}$, and the message part as $\tilde{x}$.

	      We then apply the algorithm in \cref{lem:alignment} to find an alignment of $\tilde{x}$ based on $\sigma$ and $\tilde{\sigma}$. We then write the ``synchronized'' version of $x$ to $\hat{x}$. Note that there may be both substitution errors and erasures in $\hat{x}$.

	      If $\sigma$ happens to be $\varepsilon$-self-matching, then by
	      \cref{lem:alignment} the number of half-errors in $\hat{x}$ is at most
	      $t + 6n\sqrt{\varepsilon}$, which the AG code in \cref{lem:ag-code} corrects via mixed erasure-error decoding.

	      Finally, we remove the mask $\prgmsg(s_2)$ and $\shiftmsg$ and run the $\eccmsg$ decoder:
	      $$\hat{m} = \eccmsg^{-1}(\hat{x} - \prgmsg(s_2) - \shiftmsg),$$
	      and output $\bot$ if decoding fails. Note that the subtraction are only done for the non-erased positions, so that the erasures are preserved and handled by the $\eccmsg^{-1}$.
\end{description}

Formally, the parameter setting and pseudo-code are shown below.

\begin{breakablealgorithm}
	\caption{Seeded Payload Code via Synchronization Strings}
	\label{alg:const_rate_prc}
	\begin{algorithmic}[1]
		\Require{Security parameter $\lambda$,
			\Statex \hspace{\algorithmicindent} Message $m \in [M]$, seed $s = (s_1, s_2) \in \{0,1\}^d$
			\Statex \hspace{\algorithmicindent} Parameters $n= \Omega(\lambda)$,
			\Statex \hspace{\algorithmicindent} Public shifting vector $\shiftpk \in (\sigmsg \times \sigsyn)^n$ where $(\shiftpk)_i = ((\shiftmsg)_i, (\shiftsyn)_i)$,
			\Statex \hspace{\algorithmicindent} Building blocks: Hamming ECC $\eccmsg$ and PRGs $(\prgmsg, \prgsyn)$,
			\Statex \hspace{\algorithmicindent} Alphabet $\Sigma = \sigmsg \times \sigsyn$ where $|\sigsyn|$ satisfies \cref{thm:const_rate_prc}.}

		\Statex
		\Function{$\eccload$}{$m,s,\shiftpk$}
		\State Parse $\shiftpk$ into $(\shiftmsg,\shiftsyn)$ s.t. $(\shiftpk)_i=((\shiftmsg)_i,(\shiftsyn)_i)$.
		\State Parse $s \gets (s_1, s_2)$.
		\State $\sigma \leftarrow \prgsyn(s_1)+  \shiftsyn  \in \sigsyn^{n}$
		\State $x \leftarrow \eccmsg(m)+ \prgmsg(s_2) + \shiftmsg \in \sigmsg^{n}$.
		\State Assemble $y \gets (y_1, \dots, y_{n})$ where $y_i = (x_i, \sigma_i) \in \sigmsg \times \sigsyn$.
		\State \Return $y$
		\EndFunction

		\Statex
		\Function{$\eccload^{-1}$}{$\tilde{y}\in \Sigma^{\tilde{n}},s,\shiftpk$}

		\Statex \Comment{Step 1: Recover the Synchronizing String}
		\State Parse $\shiftpk$ into $(\shiftmsg,\shiftsyn)$ s.t. $(\shiftpk)_i=((\shiftmsg)_i,(\shiftsyn)_i)$.
		\State Parse $s \gets (s_1, s_2)$.
		\State $\sigma \leftarrow \prgsyn(s_1) + \shiftsyn$.

		\Statex \Comment{Step 2: Synchronization}
		\State Decompose $\tilde{y}$ into $\tilde{\sigma}\in \sigsyn^{\tilde{n}}, \tilde{x}\in\sigmsg^{\tilde{n}}$, such that $\tilde{y}_i=(\tilde{x}_i, \tilde{\sigma}_i)$
		\State $\hat{x} \in (\sigmsg \cup \{\bot\})^n \gets$ Algorithm of \cref{lem:alignment} applied to $(\sigma, \tilde\sigma, \tilde x)$.

		\Statex \Comment{Step 3: Message Reconstruction}
		\State \Return $\hat{m} = \eccmsg^{-1}(\hat{x} - \prgmsg(s_2) - \shiftmsg)$. \textbf{if} fails \textbf{return} $\perp$.
		\EndFunction
	\end{algorithmic}
\end{breakablealgorithm}

\subsection{Analysis of \cref{alg:const_rate_prc}}\label{sec:analysis_const_rate_prc}

Now we are ready to prove the main theorem.

\thmconstrateprc*

\begin{proof}

	The uniformity of our construction will be used in later proofs.
	\paragraph{Uniformity.}

	This property is straightforward since for fixed $m$ and $s$, $\shiftmsg$ uniform makes $x$ uniform over $\sigmsg^n$, and $\shiftsyn$ uniform makes $\sigma$ uniform over $\sigsyn^n$; these are independent, hence y is uniform over $\Sigma^n$.

	\paragraph{Robustness.}

	Given that the seed $s$ is successfully recovered, we now show that the message $m$ can also be correctly decoded. Because $\sigma= \prgsyn(s_1) + \shiftsyn$ and $\shiftsyn$ is uniformly random, $\sigma$ is uniformly random over $\sigsyn^{n}$.

	Let $\varepsilon>0$ be a constant to be decided later. According to \cref{thm:random_self_matching}, a random string over an alphabet of size $|\sigsyn|$ is $\varepsilon$-self-matching (see \cref{def:self_matching}) with probability $1-(\frac{e^2}{\varepsilon^2 |\sigsyn|})^{n \varepsilon }$. We can therefore choose $|\sigsyn|=e^3 \varepsilon ^{-2}$ to make the failing probability $\le e^{-n \varepsilon}$ which is $\negl(\lambda)$ given that $n= \Omega(\lambda)$ and $\varepsilon =\Omega(1)$.

	Assuming that $\sigma$ is $\varepsilon$-self-matching, then by \cref{lem:alignment}, running \cite[Algorithm~2,3]{haeuplerSynchronizationStringsCodes2021} successfully align them, turning the edit errors into almost same number of half-errors while introducing at most $6\sqrt{\varepsilon} n$ additional half-errors. Therefore, after the synchronization, the number of half-errors in $\hat{x} $ will be at most:
	\begin{equation}
		\underbrace{\pload \cdot n}_{\text{\# Original Errors}}+\underbrace{6\sqrt{\varepsilon}n}_{\text{Additional Errors of \cref{lem:alignment}}}=(\pload+6\sqrt{\varepsilon})n
		\label{eq:hamming_error_bound}
	\end{equation}
    
	Therefore, for any $\pload\le\pmsg -\delta$, we can set $\varepsilon\coloneqq \frac{\delta^2}{36}, |\sigsyn|:=e^3\varepsilon^{-2}= \Omega(\frac{1}{\delta^4})$ which is still a constant. In this case, \eqref{eq:hamming_error_bound} will be bounded by $\pmsg \cdot n$ and the underlying $\eccmsg$ can correct all the remaining half-errors.

	\paragraph{Pseudorandomness.}
	Fix any public shifting vector $\shiftpk = (\shiftmsg, \shiftsyn)$ and message $m$. We show that the output of $\eccload(m, s, \shiftpk)$ is computationally indistinguishable from a uniform string over $\Sigma^n$. Since $s = (s_1, s_2)$ consists of two independent parts, we define a sequence of hybrid experiments $H_2, H_1, H_0$ as follows:

	\begin{itemize}
		\item \textbf{Hybrid $H_2$ (Real World):} The codeword $y$ is generated according to the real encoding procedure: $x = \eccmsg(m) + \prgmsg(s_2) + \shiftmsg$ and $\sigma = \prgsyn(s_1) + \shiftsyn$.
		\item \textbf{Hybrid $H_1$ (Uniform $\prgsyn$):} Same as $H_2$, except that the pseudo-random string $\prgsyn(s_1)$ is replaced by a truly uniform string $U_{\syn} \in \Sigma_{\syn}^n$.
		\item \textbf{Hybrid $H_0$ (Ideal World):} Same as $H_1$, except that the pseudo-random string $\prgmsg(s_2)$ is replaced by a truly uniform string $U_{\msg} \in \Sigma_{\msg}^n$.
	\end{itemize}

	The indistinguishability between $H_2$ and $H_1$ follows from the security of $\prgsyn$. Formally, if there exists a PPT distinguisher $\sA$ that distinguishes $H_2$ and $H_1$ with non-negligible advantage (even given $m$ and $\shiftpk$), we can construct a PPT adversary $\sB$ to break the PRG. $\sB$ receives a challenge string $z$ (which is either $\prgsyn(s_1)$ or a uniform string), samples an independent seed $s_2$ to compute $x = \eccmsg(m) + \prgmsg(s_2) + \shiftmsg$, and sets $\sigma = z + \shiftsyn$. Since $\sB$ can perfectly simulate the distribution of $H_2$ (if $z$ is pseudo-random) or $H_1$ (if $z$ is uniform), it inherits the advantage of $\sA$.

	The transition from $H_1$ to $H_0$ follows by an identical reduction to the security of $\prgmsg$, noting that $s_2$ is independent of $s_1$.

	In $H_0$, both $x = \eccmsg(m) + U_{\msg} + \shiftmsg$ and $\sigma = U_{\syn} + \shiftsyn$ are perfectly uniform and independent strings over their respective alphabets. Thus, the assembled sequence $y$ in $H_0$ is distributed according to $U_n$.

	Consequently, for any fixed $\shiftpk$, the oracle $\eccload(\cdot, U_d, \shiftpk)$ is computationally indistinguishable from an oracle that returns independent uniform samples from $\Sigma^n$.
\end{proof}

\subsection{Instantiating the Generic Reduction}
\label{sec:LAR_put_together}

We apply the AG code, which supports decoding from a combination of erasures and substitution errors:

\begin{lemma}[AG codes]\label{lem:ag-code}
	For every sufficiently small constant $\varepsilon>0$, there exists a square prime power $q=\poly(1/\varepsilon)$ and an explicit family of algebraic-geometry codes $C_{\mathrm{AG}} \subseteq \mathbb{F}_q^N$ such that, for all sufficiently large $N$, the code has rate $R$ and relative distance $\delta$ satisfying $R+\delta \ge 1-\varepsilon$.

	Moreover, $C_{\mathrm{AG}}$ admits polynomial-time encoding and
	polynomial-time unique decoding from any combination of $e$
	substitution errors and $s$ erasures provided
	$2e + s < d$ \cite{sakataFastErasureanderrorDecoding1998}, where $d = \delta N$ denotes the minimum distance.
\end{lemma}

Combining \cref{lem:ag-code} and \cref{thm:const_rate_prc}, we obtain:

\begin{corollary}\label{lem:payload-good-bound}
	For every constant $\varepsilon>0$ and every sufficiently small constant $\eta\in(0,\varepsilon)$, there exists a seeded payload code
	$\eccload:[M]\times\{0,1\}^d\times\{0,1\}^r\to\Sigma^n$
	with information rate at least $1-\varepsilon$ and alphabet size $|\Sigma|=\poly(1/\eta)$, that is robust to any $\varepsilon-\eta$ fraction of edit errors.
\end{corollary}

\begin{proof}
	Apply \cref{lem:ag-code} with parameter $\eta/4$
	(i.e., taking $\varepsilon = \eta/4$ in the lemma).
	Let $\eccmsg$ be the resulting AG code over $\sigmsg = \mathbb{F}_q$, with
	half-error tolerance $\pmsg = \varepsilon - \eta/2$ and
	rate $\rmsg \ge 1 - \varepsilon + \eta/4$. The alphabet size is $q = |\sigmsg| = \poly(1/\eta)$.

	Next, we instantiate the synchronization string reduction
	(see \cref{thm:const_rate_prc}) with $\eccmsg$ as the outer code. Choose the synchronization alphabet $\sigsyn$ of size $\poly(1/\eta)$
	such that the induced seeded payload code corrects every $\pload$-bounded
	edit channel whenever $\pload \le \pmsg - \eta/2$.

	Define the final alphabet by $\Sigma:=\sigmsg\times\sigsyn$. Then
	$|\Sigma|=q\cdot |\sigsyn|=\poly(1/\eta)$.
	Moreover, the information rate of the resulting seeded payload code is
	\[
		\rload=\rmsg\cdot \frac{\log q}{\log q+\log |\sigsyn|}.
	\]
	Since $|\sigsyn|=\poly(1/\eta)$, we have $\log|\sigsyn| = O(\log(1/\eta))$.
	Setting $q=(1/\eta)^C$ for a sufficiently large constant $C$ then yields
	$\rload\ge 1-\varepsilon$.

	Finally, the outer AG code is chosen to tolerate a $\pmsg=\varepsilon-\eta/2$ fraction of half-errors. Hence, by the synchronization reduction, the resulting seeded payload code corrects every
	$\pload=\pmsg-\eta/2=\varepsilon-\eta$
	fraction of edit errors.

	Therefore the resulting seeded payload code has information rate at least $1-\varepsilon$, alphabet size $|\Sigma|=\poly(1/\eta)$, and is robust against every $(\varepsilon-\eta)$-bounded edit channel.
\end{proof}

Combining \cref{lem:payload-good-bound} and \cref{thm:const-reduction}, we get the following:

\begin{corollary}[Multi-bit to zero-bit reduction]\label{cor:multi_zero_redu_sync}
	For every sufficiently small constant $\varepsilon>0$ and every alphabet $\Sigma$ with $|\Sigma| \ge \poly(1/\varepsilon)$, the following holds.

	If there exists a zero-bit public-key (resp., secret-key) $\prcz$ over $\Sigma$ that is robust against every $p_0$-bounded edit channel (resp., sublinear polynomial edit channel), then there exists a public-key (resp., secret-key) PRC over the same alphabet $\Sigma$ with rate $1-\varepsilon$ that is robust against every $\pload\le (1-\delta)\min(\frac{p_0}{4},\varepsilon)$-bounded edit channel (resp., sublinear polynomial edit channel) for any constant $\delta>0$.
\end{corollary}

Plugging in the zero-bit PRC constructed in \cref{thm: edit_PRC_zero_bit}, we get the following multi-bit PRC under \cref{asp:CG24_comb}.

\begin{corollary}\label{cor:body-rateone-result1}
	Under \Cref{asp:CG24_comb}, for every security parameter $\lambda>0$ and every sufficiently small $\varepsilon>0$, there exists a public-key PRC over alphabet of size $\poly(1/\varepsilon)$ with rate $1-\varepsilon$ that is robust against every sublinear polynomial edit channel.
\end{corollary}

Moreover, we can instantiate our generic reduction using the public-key zero-bit PRC construction of Golowich and Moitra \cite{golowich2024edit}. Their construction achieves robustness against a constant fraction of edit errors over a $\poly(\lambda)$-sized alphabet by employing an ``indexing'' approach, where each symbol is interpreted as an index into a binary Hamming-robust PRC. By plugging in this zero-bit PRC  construction, we obtain a public-key PRC that handles constant-fraction edit errors, while maintaining a rate arbitrarily close to 1.

\begin{corollary}\label{cor:body-rateone-result2}
	Under the same cryptographic assumptions as~\cite{golowich2024edit}
	(local weak PRFs), for every security parameter $\lambda>0$ and every sufficiently small $\varepsilon>0$ and $\eta \in (0,\varepsilon) $, there exists a public-key PRC over alphabet of size $\poly(\lambda,1/\eta)$ with rate $1-\varepsilon$ that is robust against every $(\varepsilon-\eta)$-bounded edit channel.
\end{corollary}

\section{Binary PRCs with Rate Close to 1/2}\label{sec:binary}

In this section, we focus on the binary setting and show how to obtain pseudorandom codes with rate arbitrarily close to $1/2$. Our construction proceeds in two steps. We first build a high-rate binary insertion-deletion code with random encoding, where the encoding randomness determines a family of binary inner codes and simultaneously guarantees pseudorandomness. We then combine this seeded binary payload code with the generic reduction from the previous section to obtain the final binary PRC.

\subsection{Seeded High-Rate Binary Insdel Code}\label{subsec:base-bianry}

We now describe the main building block for the binary reduction. Given a message $m$ and a seed $r$, the encoder outputs a binary codeword $C(m;r)$. Following the general construction of \cite{cheng2023linear}, let $q=\Theta(\log n)$ be a prime power and let $\Sigma=\mathbb F_q$, and we take $C_{\mathrm{out}}:\Sigma^K\to \Sigma^n$ to be an $[n,K,d]_q$ Reed--Solomon code, and $d=\delta n$, where $\delta=\gamma/2$. We also construct $n$ distinct binary inner codes $C_{\mathrm{in}}^1,\dots,C_{\mathrm{in}}^n$, each of codeword length $n_1$ and message length $k_1=\frac{1}{2}(1-\gamma')n_1$. The resulting binary code therefore has codeword length $N:=nn_1$.

For each seed $r\in\{0,1\}^t$, let $\mathcal C_r=(C_{r,1},C_{r,2},\dots,C_{r,n})$ denote the collection of binary inner encoders specified by $r$, where each $C_{r,i}\subseteq \{0,1\}^{n_1}$ is the inner code associated with position $i$.

Given a message $m\in \Sigma^K$, write $C_{\mathrm{out}}(m)=(a_1,a_2,\dots,a_n)\in\Sigma^n$ for the corresponding outer codeword. We then define the binary encoding of $m$ with seed $r$ by
$
C(m;r) := C_{r,1}(a_1)\circ C_{r,2}(a_2)\circ \cdots \circ C_{r,n}(a_n) \in \{0,1\}^N.
$

The decoding guarantee of \cite{cheng2023linear} relies on the following combinatorial property of the family of inner codes.

\begin{property}[with parameters $s,t,d'$]\label{prop:insd}
Let $\mathcal C_r=(C_{r,1},\dots,C_{r,n})$ be a family of binary inner codes associated with a seed $r$, where each $C_{r,i}\subseteq \{0,1\}^{n_1}$. For every $i\in[n-t]$, $j\in[n-t-1]$, every
$
w=w_i\circ w_{i+1}\circ\cdots\circ w_{i+t-1} \in \bigcirc_{\ell\in[t]} C_{r,\,i-1+\ell},
$
and every
$
u=u_j\circ u_{j+1}\circ\cdots\circ u_{j+t} \in \bigcirc_{\ell\in[t+1]} C_{r,\,j-1+\ell},
$
it holds that for every substring $w'$ of $u$, we have
$
\ed(w,w')\ge d',
$
provided that the number of unique blocks in $w$ or $u$ is at least $s$. A block $w_{i'}$ of $w$ is called \emph{unique} if $w_{i'}\neq 0^{n_1}$ and, among all blocks of $u$ coming from the same inner code $C_{r,i'}$, none is equal to $w_{i'}$. The definition of a unique block of $u$ is symmetric.
\end{property}

The following two lemmas from \cite{cheng2023linear} provide the key ingredients: for a random seed $r$, the induced family of inner codes $\mathcal C_r$ satisfies \Cref{prop:insd} with high probability, and any family satisfying this property yields a high-rate binary insertion-deletion code with efficient decoding.

\begin{lemma}[{\cite[Lemma~9]{cheng2023linear}}]\label{lem:random-inner-property}
Fix any constants $\gamma',\delta\in(0,1/2)$, and let $n_1$ be the codeword length of each inner code. Suppose that for each $i\in[n]$, the inner code $C_{r,i}\subseteq\{0,1\}^{n_1}$ is an independent random binary linear code of message rate $1/2-\gamma'$. Then, with probability at least $1-1/\poly(n)$ over the choice of $r$, the family $\mathcal C_r=(C_{r,1},\dots,C_{r,n})$ satisfies \Cref{prop:insd} with parameters $(s,t,d')$, where $\frac{d'}{n_1}=\frac{\gamma'}{10}$, $t=\frac{10\log(1/\gamma')}{\gamma'\delta}$, and $s=\frac{\delta t}{4}$.
\end{lemma}

\begin{lemma}[{\cite[Lemma~14]{cheng2023linear}}]\label{lem:prop2-rate-decode}
Let $C_{\mathrm{out}}$ be an $[n,K,d]_q$ Reed--Solomon code of relative distance $\delta$, and let $\mathcal C_r=(C_{r,1},\dots,C_{r,n})$ be a family of binary inner codes, each of codeword length $n_1$ and rate $1/2-\gamma'$. Suppose that $\mathcal C_r$ satisfies \Cref{prop:insd} with parameters $(s,t,d')$. Then the resulting concatenated binary code $C(\cdot;r)$ has rate at least $1/2-\gamma$, relative distance $\Omega\bigl(\gamma^3/\log(1/\gamma)\bigr)$, and admits a polynomial-time decoder that corrects an $\eta$-fraction of edit errors, where $\eta=\Omega\bigl(\gamma^3/\log(1/\gamma)\bigr)$.
\end{lemma}

In particular, when $r$ is chosen uniformly at random, the resulting code is decodable with probability at least $1-1/\poly(n)$ over the choice of the inner codes. This is already sufficient for existential arguments and derandomization-based constructions. However, it is not strong enough for our later pseudorandom coding application, where we need the decoding guarantee to hold with overwhelmingly high probability over the encoding randomness. To overcome this issue, we introduce an additional layer of concatenation. The next subsection explains how this extra layer amplifies the robustness guarantee while preserving the desired rate.
\subsection{Amplification of Robustness}
We now explain how to amplify the success probability of the seeded decoding guarantee from $1-1/\poly(n)$ (or even from an arbitrary fixed constant bounded away from $1$) to $1-\negl(n)$. The main idea is to introduce one additional layer of concatenation. For clarity and greater generality, we present this amplification step over an arbitrary alphabet $\Sigma$ of size $q$. Our eventual application is to the binary setting, which is then obtained as a direct specialization.

Before stating the main amplification theorem, we first prove a lemma that captures the key rejection property needed in the analysis.
\begin{lemma}\label{lem:random-reject-robust}
Let $\Sigma$ be an alphabet of size $q\ge 2$, and let $C\subseteq \Sigma^n$ be a code with minimum edit distance larger than $2\delta n$. Fix constants $b,r_1\in(0,1)$, and suppose that
\begin{equation}\label{eq:param-cond-1}
2r_1 \le 1-2b
\end{equation}
and
\begin{equation}\label{eq:param-cond-2}
2H_2(r_1\delta)+r_1\delta \log_2 q < b\delta \log_2 q.
\end{equation}
Let $\dec$ be the threshold decoder that outputs a codeword $c\in C$ iff
$\ed(y,c)\le r_1\delta n$,
and otherwise outputs $\bot$. Fix an arbitrary word $x\in \Sigma^n$, fix a set $S\subseteq [n]$ of size
\[
|S|=m:=\lfloor b\delta n\rfloor,
\]
and let $Y\in \Sigma^n$ be obtained from $x$ by replacing the symbols in $S$ independently and uniformly from $\Sigma$.

Then, with probability at least $1-2^{-\Omega(n)}$, the following stronger statement holds:
for every word $Y'$ satisfying $\ed(Y,Y')\le r_1\delta n$,
we have $\dec(Y')=\bot$.
That is,
\[
\Pr\Bigl[\exists Y' \text{ with } \ed(Y,Y')\le r_1\delta n \text{ and } \dec(Y')\neq \bot\Bigr]
\le 2^{-\Omega(n)}.
\]
\end{lemma}

\begin{proof}
Let $B := \{y \in\Sigma^n: y_i = x_i$ for all $i\notin S \}$. Then $Y$ is uniform over $B$, and $|B| = q^m$. Define 
\[
\mathrm{Ball}_E(c,r)
:=
\{y\in\Sigma^n : \ed(c,y)\le rn\}.
\]
Let
\[
C_1:=\{c\in C: B\cap \mathrm{Ball}_E(c,r_1\delta)\neq \emptyset\}.
\]
\begin{claim}
\[
\bigl|\mathrm{Ball}_E(c,r)\bigr|
\le
\poly(n)\cdot 2^{ n\left(2H_2(r/2)+\frac r2\log_2 q\right)}
\]
\end{claim}
\begin{proof}
For a fixed $j$, the number of strings obtained from $c$ by deleting $j$ positions is at most $\binom{n}{j}$, and each resulting string has length $n-j$. From any fixed string $z \in \Sigma^{n-j}$, the number of length-$n$ strings obtained by inserting $j$ symbols is at most $\binom{n}{j}q^j$. Summing over all $0\le j\le \lfloor rn/2\rfloor$ yields
\[
\begin{aligned}
\bigl|\mathrm{Ball}_E(c,r)\bigr|& \le \sum_{j=0}^{\lfloor rn/2\rfloor}\binom{n}{j}^2 q^j\\
&\le \poly(n)\cdot 2^{ n\left(2H_2(r/2)+\frac r2\log_2 q\right)}.\\
\end{aligned}
\]
\end{proof}

\begin{claim}
Fix a codeword $c\in C$. Suppose there exists a word $c'\in B$ such that $\ed(c,c')\le r_1\delta n$. Then
\[
\bigl|B\cap \mathrm{Ball}_E(c,\delta)\bigr| \ge \sum_{i=0}^{\lfloor b\delta n\rfloor}\binom{m}{i}(q-1)^i\geq 2^{b\delta n\log_2 q-O(1)}.
\]
\end{claim}
\begin{proof}
By Equation~\eqref{eq:param-cond-1}, every word in $B$ obtained from $c'$ by changing at most $\lfloor b\delta n\rfloor$ free coordinates lies in $\mathrm{Ball}_E(c,\delta)$.
\end{proof}
Therefore
\[
\frac{|\mathrm{Ball}_E(c,2r_1\delta)|} {\bigl|B\cap \mathrm{Ball}_E(c,\delta)\bigr|} \le \poly(n)\cdot
2^{n\left(2H_2(r_1\delta)+r_1\delta\log_2 q-b\delta\log_2 q\right)+O(1)}.
\]
By Equation~\eqref{eq:param-cond-2}, the right-hand side is at most $2^{-\Omega(n)}$.

Now, if $c_1\neq c_2$ are distinct codewords, then
\[
\bigl(B\cap \mathrm{Ball}_E(c_1,\delta)\bigr) \cap \bigl(B\cap \mathrm{Ball}_E(c_2,\delta)\bigr) =\emptyset,
\]
Therefore
\[
|B| \ge \sum_{c\in C_1}\bigl|B\cap \mathrm{Ball}_E(c,\delta)\bigr|,
\]
and thus
\[
\begin{aligned}
\Pr\bigl[\exists c\in C:\ \ed(Y,c)\le 2r_1\delta n\bigr] &\le \frac{\sum_{c\in C_1}|B\cap \mathrm{Ball}_E(c,2r_1\delta)|}{|B|} \\&\le \frac{\sum_{c\in C_1}|\mathrm{Ball}_E(c,2r_1\delta)|}{\sum_{c\in C_1}\bigl|B\cap \mathrm{Ball}_E(c,\delta)\bigr|} \le 2^{-\Omega(n)}.
\end{aligned}
\]
Fix such a $Y$, and let $Y'$ be any word with $\ed(Y,Y')\le r_1\delta n$. Then for every $c\in C$, by the triangle inequality, $\ed(Y',c)\ge \ed(Y,c)-\ed(Y,Y')>2r_1\delta n-r_1\delta n=r_1\delta n$.
Thus $Y'$ is outside the decoding radius of every codeword, and therefore $\dec(Y')=\bot$.

This proves that, with probability at least $1-2^{-\Omega(n)}$ over $Y$, every further edit corruption of size at most $r_1\delta n$ is still rejected by the decoder.
\end{proof}

We are now ready to state the main amplification theorem. It shows that any seeded insertion-deletion code satisfying the properties above can be boosted, via one additional layer of concatenation, so that the decoding guarantee holds with overwhelming probability over the seed.

\begin{theorem}\label{thm:amplification}
Let $\Sigma$ be an alphabet of size $q\ge 2$. Suppose we have an encoding map $C:\Sigma^K\times \Sigma^t\to \Sigma^N$, such that the following hold.
\begin{enumerate}
    \item For every message $m\in\Sigma^K$,
    \[
    C(m;U_t)=U_N.
    \]
    \item For every seed $r\in\Sigma^t$, the code $\{C(m;r):m\in\Sigma^K\}\subseteq \Sigma^N$    has minimum edit distance larger than $2\delta N$.
    \item There exists a polynomial-time decoder $\dec_{\mathrm{in}}$ and a constant $r_1\in(0,1)$ such that, with probability at least $1-\rho$ over $r\sim U_t$, the following holds simultaneously for every message $m\in\Sigma^K$ and every word $z\in\Sigma^N$:
    \[
    \dec_{\mathrm{in}}(r,z)=m \quad\Longleftrightarrow\quad \ed\bigl(z,C(m;r)\bigr)\le r_1\delta N.
    \]
    \item There exists $b\in(0,1)$ such that
    \begin{equation}\label{eq:base-param-1}
    2r_1\le 1-2b,
    \end{equation}
    \begin{equation}\label{eq:base-param-2}
    2H_2(r_1\delta)+r_1\delta\log_2 q<b\delta\log_2 q,
    \end{equation}
    and
    \begin{equation}\label{eq:base-param-3}
    3b+\frac{3}{2}r_1<\frac12.
    \end{equation}
\end{enumerate}
Let $E_{\mathrm{out}}:(\Sigma^K)^{M'}\to (\Sigma^K)^{N'}$ be an outer code of relative Hamming distance $\Delta_{\mathrm{out}}>0$, equipped with a polynomial-time error-erasure decoder. Let
$\tau>0$ satisfy
$\rho+\frac{6\tau}{r_1}<\Delta_{\mathrm{out}}$.
Assume also that $N =  \omega(\log N')$. Then there is a concatenated encoding map
\[
C':(\Sigma^K)^{M'}\times (\Sigma^t)^{N'}\to \Sigma^{N'N}
\]
such that:

\begin{enumerate}
    \item For every message $m'\in(\Sigma^K)^{M'}$,
    \[
    C'(m';U_{N't})=U_{N'N}.
    \]

    \item There exists a polynomial-time decoder $\dec_{\mathrm{out}}$ such that, for every message $m'\in(\Sigma^K)^{M'}$, with probability at least $1-2^{-\Omega(N')}$ over $r'\sim U_{N't}$, the following holds: for every received word $y$,
    \[
    \ed\bigl(y,C'(m';r')\bigr)\le \tau\delta N'N
    \quad\Longrightarrow\quad
    \dec_{\mathrm{out}}(r',y)=m'.
    \]
\end{enumerate}

\end{theorem}

We now describe the sequential decoding procedure. For convenience, define $L:=\left(b\delta+\frac{r_1\delta}{2}\right)N$. The decoder maintains a current word $Y^{(i)}$, where initially $Y^{(1)}:=y$. For $i=1,2,\cdots,N'$, given the current word $Y^{(i)}$, the decoder proceeds as follows.

\begin{enumerate}
    \item Slide a window of length $N$ from left to right over $Y^{(i)}$. Equivalently, for each contiguous interval $J\subseteq Y^{(i)}$ of length exactly $N$, run the inner decoder $\dec_{\mathrm{in}}(r_i,J)$.

    \item If there exists such an interval $J\subseteq Y^{(i)}$ with $\dec_{\inn}(r_i,J)\neq \bot$, let $J_i$ be the leftmost such interval, and set $\widehat u_i:=\dec_{\inn}(r_i,J_i)$. We then declare block $i$ to be successfully decoded with value $\widehat u_i$.

    \item After a successful decoding, update the current word by keeping only the first $L$ symbols and the last $L$ symbols of $J_i$, deleting the middle part of $J_i$, and retaining all symbols outside $J_i$. The resulting word is denoted by $Y^{(i+1)}$.

    \item If no interval $J\subseteq Y^{(i)}$ of length $N$ satisfies $\dec_{\inn}(r_i,J)\neq \bot$, then set $\widehat u_i:=\bot$ and $Y^{(i+1)}:=Y^{(i)}$.
\end{enumerate}

At the end, this produces
\[
(\widehat u_1,\cdots,\widehat u_{N'})\in(\Sigma^K\cup\{\bot\})^{N'},
\]
which we feed to the outer decoder.

\begin{proof}
We state the condition with the stronger margin $\rho+\frac{6\tau}{r_1}<\Delta_{\mathrm{out}}$ to absorb the concentration slack and the union-bound losses appearing in the proof. Let $(u_1,\cdots,u_{N'})=E_{\mathrm{out}}(m')$, $x_i:=C(u_i;r_i)\in\Sigma^N$, and $x:=x_1\circ\cdots\circ x_{N'}=C'(m';r')$. Assume \[\ed(y,x)\le \tau\delta N'N.\] Fix an edit sequence from $x$ to $y$. Charge each edit to one of the $N'$ blocks, charging an insertion between two consecutive blocks to the earlier block. Let $e_i$ be the total number of edits charged to block $i$. Then
\[
\sum_{i=1}^{N'} e_i\le \tau\delta N'N.
\]
Call block $i$ \emph{locally good} if
\[
e_i\le \frac{r_1\delta N}{2}.
\]
Then the number of blocks that are not locally good is at most $\frac{\tau\delta N'N}{(r_1\delta N)/2} = \frac{2\tau}{r_1}N'$.
Hence at least $\left(1-\frac{2\tau}{r_1}\right)N'$ blocks are locally good. Call a seed $r_i$ \emph{good} if it satisfies item~(3). Since the $r_i$'s are independent and each is good with probability at least $1-\rho$, Hoeffding's inequality implies that, with probability at least $1-2^{-\Omega(N')}$, at most
$(\rho+\varepsilon)N'$ seeds are bad, where $\varepsilon>0$ is an arbitrarily small fixed constant. Fix such a realization of $r'$.

\begin{claim}[Good blocks admit a correct window]\label{claim:good-window}
Suppose block $i$ is locally good, and suppose the portion of the received word corresponding to block $i$ is still present in $Y^{(i)}$. If $r_i$ is good, then there exists a contiguous interval $J\subseteq Y^{(i)}$ of length exactly $N$ such that
\[
\dec_{\mathrm{in}}(r_i,J)=u_i.
\]
In particular, the $i$-th step of the sequential decoder has at least one successful window.
\end{claim}
\begin{proof}
Since block $i$ is locally good, we have $e_i\le \frac{r_1\delta N}{2}$.
Hence there exists a contiguous interval $J\subseteq Y^{(i)}$ of length exactly $N$ such that $\ed(J,x_i)\le 2e_i\le r_1\delta N$.
Because $r_i$ is good, the exact threshold property of the inner decoder implies $\dec_{\mathrm{in}}(r_i,J)=u_i$.
Thus $J$ is a successful window for block $i$.
\end{proof}

\begin{claim}[Wrong successful decoding consumes many edits]\label{claim:wrong-costs-edits}
Suppose $r_i$ is good. If, during the decoding of block $i$, some length-$N$ interval $J\subseteq Y^{(i)}$ is decoded to an incorrect message, then the middle portion removed in Step~3 contains at least $\frac{r_1\delta N}{2}$
charged edits.
\end{claim}
\begin{proof}
Suppose otherwise that fewer than $\frac{r_1\delta N}{2}$ charged edits occur in the removed middle interval. Since the removed middle interval has length $N-2L = N-(2b\delta+r_1\delta)N$,
it follows that at least
\[
N-2L-\frac{r_1\delta N}{2}
=
N-\left(2b+\frac{3r_1}{2}\right)\delta N
\]
symbols in this interval come from the transmitted codeword. We bound separately the number of such symbols coming from the true $i$-th block and from other blocks.

First, at most $\left(1-\frac{\delta}{2}\right)N$
of these symbols can come from the true $i$-th block. Indeed, if more than $\left(1-\frac{\delta}{2}\right)N$ symbols came from the $i$-th block, then the interval would have edit distance strictly less than $\delta N$ from $x_i$. Since $r_i$ is good and the code for seed $r_i$ has minimum edit distance greater than $2\delta N$, an interval within distance $<\delta N$ of $x_i$ cannot decode to a wrong message. Since the seeds $r_1,\dots,r_{N'}$ are independent, conditioning on all previous decoding outcomes and on the current seed $r_i$ does not change the fact that any symbols of $J$ originating from blocks other than $i$ are independent of $r_i$; moreover, by the exact-uniformity property of the inner code, those symbols are distributed as uniform random symbols relative to the code determined by $r_i$.

Second, except with probability $2^{-\Omega(N)}$, at most $b\delta N$ symbols in the interval can come from other blocks. Indeed, if more than $b\delta N$ symbols came from other blocks, then the interval would contain at least $b\delta N$ symbols that are independent random symbols relative to the key $r_i$. By \cref{lem:random-reject-robust}, such an interval is decoded to $\bot$ except with probability $2^{-\Omega(N)}$.

Therefore, except with probability $2^{-\Omega(N)}$, the total number of transmitted symbols in the removed interval is at most $\left(1-\frac{\delta}{2}\right)N+b\delta N$.
Comparing this with the lower bound above yields
\[
N-\left(2b+\frac{3r_1}{2}\right)\delta N
\le
\left(1-\frac{\delta}{2}\right)N+b\delta N.
\]
Equivalently,
\[
3b+\frac{3}{2}r_1\ge \frac12,
\]
contradicting \eqref{eq:base-param-3}. This proves the claim.
\end{proof}

\begin{claim}[Successful steps preserve future locally good blocks]\label{claim:preserve-future-good}
Suppose that, during the decoding of block $i$, the decoder finds a successful length-$N$ interval $J_i$, i.e., $\dec_{\inn}(r_i,J_i)\neq \bot$.
Then, except with probability $2^{-\Omega(N)}$, Step~3 deletes no symbol belonging to any later locally good block.
\end{claim}

\begin{proof}
Write
\[
J_i=x_1\circ x_2\circ x_3,
\]
where $x_2$ is the portion of $J_i$ coming from the true $i$-th block, while $x_1$ and $x_3$ come from neighboring blocks (possibly empty).

We claim that, except with probability $2^{-\Omega(N)}$,
$|x_1|\le L$ and $|x_3|\le L$.
Indeed, suppose $|x_1|>L$. Then, after discarding at most $r_1\delta N/2$ inserted symbols, there are still at least
$L-\frac{r_1\delta N}{2}=b\delta N$
symbols in $x_1$ coming from outside the $i$-th block. These symbols are independent random symbols relative to the key $r_i$. By \cref{lem:random-reject-robust}, any length-$N$ interval containing such a set of symbols is decoded to $\bot$ except with probability $2^{-\Omega(N)}$, contradicting the assumption that $J_i$ is successful. Thus $|x_1|\le L$ except with probability $2^{-\Omega(N)}$. The same argument applies to $x_3$.

Now consider the case where $x_2$ does not come from the true $i$-th block. Then the same random-rejection argument applies directly to $x_2$: if $x_2$ came from a locally good block different from $i$, then $J_i$ would contain at least $b\delta N$ symbols that are independent random symbols relative to $r_i$, and hence $J_i$ would decode to $\bot$ except with probability $2^{-\Omega(N)}$. Therefore, except with negligible probability, such an $x_2$ cannot come from a locally good block. By the previous claim, any successful but incorrect decoding step already consumes at least $\frac{r_1\delta N}{2}$ charged edits in the removed interval, so this case cannot destroy more than the budget already accounted for.

Therefore, except with probability $2^{-\Omega(N)}$, the middle part removed in Step~3 lies entirely inside the current block, and no symbol belonging to any later locally good block is deleted.
\end{proof}

By \cref{claim:wrong-costs-edits}, each wrong successful decoding step consumes at least $\frac{r_1\delta N}{2}$ charged edits. Since the total charged edit budget is at most $\tau\delta N'N$, the number of such steps is at most
\[
\frac{\tau\delta N'N}{(r_1\delta N)/2}
=
\frac{2\tau}{r_1}N'.
\]
The only exceptional events arise from applications of \Cref{lem:random-reject-robust}. Therefore, except with probability $2^{-\Omega(N')}$, the number of wrong successful decoding steps is at most $\frac{2\tau}{r_1}N'$. Indeed, there are at most $N'$ stages, and since $N=\omega(\log N')$, a union bound over the failure probabilities $2^{-\Omega(N)}$ contributes at most $2^{-\Omega(N')}$. Therefore, except with probability $2^{-\Omega(N')}$, at least
\[
\left(1-\rho-\varepsilon-\frac{4\tau}{r_1}\right)N'
\]
outer positions are decoded correctly. This is because a position can fail to contribute a correct outer symbol for one of the following reasons:
\begin{itemize}
    \item the block is not locally good, which happens for at most $\frac{2\tau}{r_1}N'$ positions;
    \item the seed $r_i$ is bad, which happens for at most $(\rho+\varepsilon)N'$ positions except with probability $2^{-\Omega(N')}$;
    \item the block is destroyed by an earlier wrong successful step, which happens for at most $\frac{2\tau}{r_1}N'$ positions.
\end{itemize}
By \cref{claim:good-window}, every remaining position has a correct successful window, and by \cref{claim:preserve-future-good}, correct successful steps do not destroy later locally good blocks except with probability $2^{-\Omega(N')}$. Therefore at least
\[
\left(1-\rho-\varepsilon-\frac{4\tau}{r_1}\right)N'
\]
positions are decoded correctly.

Therefore, except with probability $2^{-\Omega(N')}$, the number of incorrect or erased outer positions is at most $(\rho+\varepsilon+\frac{4\tau}{r_1})N'$.
Choosing $\varepsilon>0$ small enough so that $\rho+\varepsilon+\frac{4\tau}{r_1}<\Delta_{\mathrm{out}}$,
the outer decoder recovers $(u_1,\cdots,u_{N'})=E_{\mathrm{out}}(m')$,
and hence recovers $m'$.

Finally, for every fixed $m'$, the blocks
$C(u_1;U_t),\cdots,C(u_{N'};U_t)$
are independent uniform strings in $\Sigma^N$. Therefore their concatenation is uniform in $\Sigma^{N'N}$, that is, $C'(m';U_{N't})=U_{N'N}$.
This completes the proof.
\end{proof}

\begin{corollary}\label{cor:binary-payload}
For every sufficiently large $N$, there exists an encoding map
\[
C:[M]\times\{0,1\}^r\to\{0,1\}^N
\]
with $r=O(N\log N)$ and rate $1/2-\gamma$ such that:
\begin{enumerate}
    \item for every message $m\in[M]$,
    \[
    C(m;U_r)=U_N;
    \]
    \item there exists a polynomial-time decoder $\dec$ such that, with probability at least $1-\negl(N)$ over $s\sim U_r$, the map $m\mapsto C(m;s)$ has relative edit distance $\Omega\Bigl(\frac{\gamma^3}{\log(1/\gamma)}\Bigr)$,
    and moreover, for every message $m\in[M]$ and every received word $y$,
    \[
    \ed\bigl(y,C(m;s)\bigr)
    \le
    \Omega\Bigl(\frac{\gamma^3}{\log(1/\gamma)}\Bigr)\cdot N
    \quad\Longrightarrow\quad
    \dec(s,y)=m.
    \]
\end{enumerate}
\end{corollary}
\begin{proof}

We instantiate \Cref{thm:amplification} using the seeded binary code from \Cref{subsec:base-bianry} as the inner code.

Let
\[
C_{\mathrm{in}}:[K]\times\{0,1\}^{r_0}\to\{0,1\}^{N_0}
\]
be the inner binary encoding map, with codeword length $N_0$, message space $[K]$, seed length $r_0$, and decoding radius $\eta=\Omega\Bigl(\frac{\gamma^3}{\log(1/\gamma)}\Bigr)$.
By \Cref{lem:random-inner-property}, the bad-seed probability is $\rho=\frac{1}{\poly(N_0)}$.

We set inner-code distance parameter with $\delta:=\eta$. Choose an outer Reed--Solomon code of rate $1-\varepsilon$, where $\varepsilon>0$ is a sufficiently small constant, so that its relative Hamming distance satisfies $\Delta_{\mathrm{out}}=\varepsilon+o(1)$.
Next choose a constant $\tau>0$ such that $\rho+\frac{6\tau}{r_1}<\Delta_{\mathrm{out}}$. For example, for all sufficiently large $N'$, which is the length of Reed--Solomon code, it suffices to take $\tau=\frac{r_1\varepsilon}{12}$.

We now apply \Cref{thm:amplification}, which yields an amplified encoding map
\[
C:( [K])^{M'}\times (\{0,1\}^{r_0})^{N'}\to\{0,1\}^{N'N_0}
\]
with seed failure probability $2^{-\Omega(N')}$, exact uniformity under a uniform seed, and decoding radius
\[
\tau\eta N'N_0
=
\Omega\Bigl(\frac{\gamma^3}{\log(1/\gamma)}\Bigr)\cdot N'N_0.
\]

We now rename parameters by setting $M:=K^{M'}$, $r:=N'r_0$, and $N:=N'N_0$.
With this notation, the amplified code takes the form
\[
C:[M]\times\{0,1\}^r\to\{0,1\}^N,
\]
with rate $\frac{\log_2 M}{N} = \frac{M'\log_2 K}{N'N_0}$.
Since the outer Reed--Solomon code has rate $1-\varepsilon$ and the inner binary code has rate at least $1/2-\gamma$, we obtain
\[
\frac{\log_2 M}{N}\ge (1-\varepsilon)\Bigl(\frac12-\gamma\Bigr).
\]
By choosing $\varepsilon>0$ sufficiently small and absorbing constant-factor losses into $\gamma$, we may write the final rate as at least $1/2-\gamma$. Similarly, the theorem gives relative edit distance and decoding radius
\[
\Omega(\tau\eta)=\Omega(\eta)=\Omega\Bigl(\frac{\gamma^3}{\log(1/\gamma)}\Bigr),
\]
so the amplified code has the claimed asymptotic distance and correctable noise rate, up to constant factors.

It remains to verify the seed length bound and the uniformity statement. Since the seed of the amplified code consists of $N'$ independent inner seeds, each of length $r_0=O(N_0\log N_0)$, we have $r=N'r_0$. Therefore
\[
r=O(N'N_0\log N_0)=O(N\log N).
\]
Finally, exact uniformity under a uniform seed follows directly from the uniformity and independence of the inner code, and thus for every fixed message $m\in[M]$,
\[
C(m;U_r)=U_N.
\]
\end{proof}

\subsection{Binary Payload Code and Put Together}
We now package the binary construction above into the payload code (see \Cref{def:seeded_payload_code}) required by the generic reduction in \Cref{sec:generic_reduction_to_zero_bit}. Recall that the amplified binary seeded code provides an encoding map
$
C:[M]\times\{0,1\}^r\to\{0,1\}^n,
$
such that a uniformly random seed yields a uniformly random codeword, while decoding succeeds with overwhelming probability over the choice of the seed.

To fit the framework of \Cref{sec:generic_reduction_to_zero_bit}, we derive the seed of this code from two parts. A short seed
$
s\in\{0,1\}^d
$
is first expanded by a pseudorandom generator
$
\prg:\{0,1\}^d\to\{0,1\}^r,
$
and then combined with a public shifting vector
$
\shiftpk\in\{0,1\}^r.
$
We take $d=n^\beta$ for some constant $\beta>0$, which is sufficient for generating pseudorandom strings of length $r=O(n\log n)$. We therefore define the payload encoder by
\[
\eccload:[M]\times\{0,1\}^d\times\{0,1\}^r\to\{0,1\}^n,
\quad
\eccload(m,s,\shiftpk):=C\bigl(m,\prg(s)\oplus \shiftpk\bigr).
\]
The corresponding decoder is
\[
\eccload^{-1}:\{0,1\}^\ast\times\{0,1\}^d\times\{0,1\}^r\to [M]\cup\{\bot\},
\quad
\eccload^{-1}(y,s,\shiftpk):=\dec\bigl(\prg(s)\oplus \shiftpk,\; y\bigr),
\]
where $\dec$ is the decoder of the amplified binary seeded code.

The key point is that, when $\shiftpk$ is uniform, the effective seed
$
\prg(s)\oplus \shiftpk
$
is itself uniform for every fixed $s$. Hence the resulting payload codeword is exactly uniform. At the same time, once the short seed $s$ is recovered by the outer reduction, the decoder can reconstruct the effective seed and use it to decode the payload.

We now fix the parameters for the binary payload code. Given any target $\varepsilon>0$, we choose the underlying binary insertion-deletion code so that its rate is at least $1/2-\varepsilon$. By the construction above, the corresponding decoding radius remains $\Omega\left(\frac{\varepsilon^3}{\log(1/\varepsilon)}\right)$.
Indeed, this follows by taking the parameter $\gamma$ in \Cref{cor:binary-payload} to be a sufficiently small constant multiple of $\varepsilon$, and absorbing constant-factor losses into the $\Omega(\cdot)$ notation. The resulting encoding and decoding procedures are given in \Cref{alg:binary-payload}.

\begin{lemma}\label{lem:binary-payload-interface}
Suppose
\[
\eccload:[M]\times\{0,1\}^d\times\{0,1\}^r\to\{0,1\}^n
\]
is the payload encoding map defined above, where the underlying binary seeded code has rate at least $1/2-\varepsilon$ and decoding radius
\[
\delta=\Omega\left(\frac{\varepsilon^3}{\log(1/\varepsilon)}\right).
\]
Then $\eccload$ is a seeded payload code with rate $1 / 2 - \varepsilon$ that is robust to any $\delta$ fraction of edit errors.
\end{lemma}

\begin{proof}
The robustness and uniformity statements follow directly from \Cref{cor:binary-payload}. For pseudorandomness, observe that for every fixed public shifting vector $\shiftpk$, the effective seed $\prg(s)\oplus \shiftpk$ is computationally indistinguishable from uniform by the pseudorandomness of $\prg$. Since the map
\[
(s,\shiftpk)\mapsto \eccload(m,s,\shiftpk)=C\bigl(m,\prg(s)\oplus\shiftpk\bigr)
\]
is efficiently computable, computational indistinguishability is preserved under this encoding procedure. Thus the output distribution of $\eccload(\cdot,U_d,\shiftpk)$ is computationally indistinguishable from uniform. A standard hybrid argument, analogous to that used in the proof of \Cref{thm:const-reduction}, completes the proof.
\end{proof}

\begin{breakablealgorithm}
\caption{Binary Payload Code with Rate Close to $1/2$}\label{alg:binary-payload}
\begin{algorithmic}[1]
    \Require Security parameter $\lambda$, 
    \Statex
    \hspace{\algorithmicindent}Message $m\in[M]$, seed $s\in\{0,1\}^{d}$, public shifting vector $\shiftpk\in\{0,1\}^{r}$
    \Statex\hspace{\algorithmicindent}Pseudorandom generator $\prg:\{0,1\}^{d}\to\{0,1\}^{r}$
    \Statex \hspace{\algorithmicindent}Outer Reed--Solomon code $E_{\mathrm{out}}:(\{0,1\}^{K})^{M'}\to(\{0,1\}^{K})^{N'}$ with decoder $\dec_{\mathrm{RS}}$
    \Statex \hspace{\algorithmicindent}Base binary seeded code $C_{\mathrm{base}}:\{0,1\}^{K}\times\{0,1\}^{t_0}\to\{0,1\}^{N_0}$ with decoder $\dec_{\mathrm{base}}$
    \Statex \hspace{\algorithmicindent}Constant $L$

    \Statex

    \Function{$\eccload$}{$m,s,\shiftpk$}
        \State Compute the effective seed
        $r \gets \prg(s)\oplus \shiftpk$
        \State Parse $m=(m_1,\dots,m_{M'})\in(\{0,1\}^{K})^{M'}$
        \State Parse $r=(r_1,\dots,r_{N'})\in(\{0,1\}^{t_0})^{N'}$
        \State Compute the outer codeword
        $(u_1,\dots,u_{N'})\gets E_{\mathrm{out}}(m)\in(\{0,1\}^{K})^{N'}$
        \For{$i=1$ to $N'$}
            \State $x_i\gets C_{\mathrm{base}}(u_i;r_i)\in\{0,1\}^{N_0}$
        \EndFor
        \State \Return
        $x_1\circ x_2\circ\cdots\circ x_{N'}\in\{0,1\}^{N'N_0}$
    \EndFunction

    \Statex

    \Function{$\eccload^{-1}$}{$y,s,\shiftpk$}
        \State Compute the effective seed
        $r \gets \prg(s)\oplus \shiftpk$
        \State Parse $r=(r_1,\dots,r_{N'})\in(\{0,1\}^{t_0})^{N'}$
        \State Initialize $Y^{(1)}\gets y$

        \For{$i=1$ to $N'$}
            \State $\widehat{u}_i\gets\bot$
            \For{every contiguous interval $J\subseteq Y^{(i)}$ of length exactly $N_0$}
                \If{$\dec_{\mathrm{base}}(r_i,J)\neq\bot$}
                    \State Let $J_i$ be the leftmost such interval\Comment{Choose the first successful window}
                    \State $\widehat{u}_i\gets \dec_{\mathrm{base}}(r_i,J_i)$
                    \State \textbf{break}
                \EndIf
            \EndFor

            \If{$\widehat{u}_i\neq\bot$}
                \State Write $J_i=J_i^{\mathrm{left}}\circ J_i^{\mathrm{mid}}\circ J_i^{\mathrm{right}}$
                \Statex \hspace{\algorithmicindent}where $|J_i^{\mathrm{left}}|=|J_i^{\mathrm{right}}|=L$
                \State Delete the middle block $J_i^{\mathrm{mid}}$ from $Y^{(i)}$
                \State Denote the resulting word by $Y^{(i+1)}$
            \Else
                \State $Y^{(i+1)}\gets Y^{(i)}$\Comment{No successful window found for this block}
            \EndIf
        \EndFor

        \State Run the outer decoder on $(\widehat{u}_1,\dots,\widehat{u}_{N'})\in(\{0,1\}^{K}\cup\{\bot\})^{N'}$
        \State \Return $\dec_{\mathrm{RS}}(\widehat{u}_1,\dots,\widehat{u}_{N'})$
    \EndFunction
\end{algorithmic}
\end{breakablealgorithm}

\paragraph{Instantiation from the binary payload code.}

We now combine the binary payload code constructed above with the generic reduction in \Cref{alg:const-rate}. Using the guarantees provided by \Cref{cor:const-reduction-fix-seed-code}, we obtain the following binary pseudorandom code.

\begin{theorem}[Binary PRC with rate close to $1/2$]\label{thm:body-binary-reduction}
For every sufficiently small constant $\varepsilon>0$, the following hold.

\begin{enumerate}
    \item
    Let $p\in(0,1)$ be any constant. If there exists a zero-bit public-key (resp., secret-key) binary PRC that is robust against every $p$-bounded edit channel, then there exists a public-key (resp., secret-key) binary PRC with rate $1/2-\varepsilon$ that is robust against every $p'$-bounded edit channel, where
    \[
    p'=\Omega\left(\min\left\{p,\frac{\varepsilon^3}{\log(1/\varepsilon)}\right\}\right).
    \]

    \item
    If there exists a zero-bit public-key (resp., secret-key) binary PRC that is robust against every sublinear polynomial edit channel, then there exists a public-key (resp., secret-key) binary PRC with rate $1/2-\varepsilon$ that is robust against every sublinear polynomial edit channel.
\end{enumerate}
\end{theorem}

Using \cref{thm: edit_PRC_zero_bit}, we get the following binary PRC under \Cref{asp:CG24_comb}.

\begin{corollary}\label{cor:body-binary-reduction}
Under \Cref{asp:CG24_comb}, for every security parameter $\lambda>0$ and every sufficiently small $\varepsilon>0$, there exists a binary public-key PRC with rate $1/2-\varepsilon$ that is robust against every sublinear polynomial edit channel.
\end{corollary}
\bibliographystyle{alpha}
\bibliography{references}

@article{schulman2002asymptotically,
  title         = {Asymptotically good codes correcting insertions, deletions, and transpositions},
  author        = {Schulman, Leonard J and Zuckerman, David},
  journal       = {IEEE transactions on information theory},
  volume        = {45},
  number        = {7},
  pages         = {2552--2557},
  year          = {2002},
  publisher     = {IEEE}
}

@article{sakataFastErasureanderrorDecoding1998,
  title         = {Fast Erasure-and-Error Decoding of Algebraic Geometry Codes up to the {{Feng-Rao}} Bound},
  author        = {Sakata, S. and Leonard, D.A. and Jensen, H.E. and Hoholdt, T.},
  year          = 1998,
  month         = jul,
  journal       = {IEEE Transactions on Information Theory},
  volume        = {44},
  number        = {4},
  pages         = {1558--1564},
  issn          = {1557-9654},
  doi           = {10.1109/18.681332},
  urldate       = {2026-04-30}
}

@article{golowich2024edit,
  title         = {Edit distance robust watermarks via indexing pseudorandom codes},
  author        = {Golowich, Noah and Moitra, Ankur},
  journal       = {Advances in Neural Information Processing Systems},
  volume        = {37},
  pages         = {20645--20693},
  year          = {2024}
}

@article{christ2025improved,
  title         = {Improved Pseudorandom Codes from Permuted Puzzles},
  author        = {Christ, Miranda and Golowich, Noah and Gunn, Sam and Moitra, Ankur and Wichs, Daniel},
  journal       = {arXiv preprint arXiv:2512.08918},
  year          = {2025}
}

@inproceedings{CGK_embedding,
  author        = {Chakraborty, Diptarka and Goldenberg, Elazar and Kouck\'{y}, Michal},
  title         = {Streaming algorithms for embedding and computing edit distance in the low distance regime},
  year          = {2016},
  isbn          = {9781450341325},
  publisher     = {Association for Computing Machinery},
  address       = {New York, NY, USA},
  url           = {https://doi.org/10.1145/2897518.2897577},
  doi           = {10.1145/2897518.2897577},
  booktitle     = {Proceedings of the Forty-Eighth Annual ACM Symposium on Theory of Computing},
  pages         = {712–725},
  numpages      = {14},
  location      = {Cambridge, MA, USA},
  series        = {STOC '16}
}

@inproceedings{christPseudorandomErrorCorrectingCodes2024,
  title         = {Pseudorandom Error-Correcting Codes},
  author        = {Christ, Miranda and Gunn, Sam},
  booktitle     = {Annual International Cryptology Conference},
  pages         = {325--347},
  year          = {2024}
}

@article{cheng2023linear,
  title         = {Linear insertion deletion codes in the high-noise and high-rate regimes},
  author        = {Cheng, Kuan and Jin, Zhengzhong and Li, Xin and Wei, Zhide and Zheng, Yu},
  journal       = {arXiv preprint arXiv:2303.17370},
  year          = {2023}
}

@inproceedings{GG25_RANDOM,
  author        = {Ghentiyala, Surendra and Guruswami, Venkatesan},
  title         = {{New Constructions of Pseudorandom Codes}},
  booktitle     = {Approximation, Randomization, and Combinatorial Optimization. Algorithms and Techniques (APPROX/RANDOM 2025)},
  pages         = {54:1--54:22},
  series        = {Leibniz International Proceedings in Informatics (LIPIcs)},
  isbn          = {978-3-95977-397-3},
  issn          = {1868-8969},
  year          = {2025},
  volume        = {353},
  editor        = {Ene, Alina and Chattopadhyay, Eshan},
  publisher     = {Schloss Dagstuhl -- Leibniz-Zentrum f{\"u}r Informatik},
  address       = {Dagstuhl, Germany},
  url           = {https://drops.dagstuhl.de/entities/document/10.4230/LIPIcs.APPROX/RANDOM.2025.54},
  urn           = {urn:nbn:de:0030-drops-244202},
  doi           = {10.4230/LIPIcs.APPROX/RANDOM.2025.54}
}

@article{haeuplerSynchronizationStringsCodes2021,
  title         = {Synchronization {{Strings}}: {{Codes}} for {{Insertions}} and {{Deletions Approaching}} the {{Singleton Bound}}},
  shorttitle    = {Synchronization {{Strings}}},
  author        = {Haeupler, Bernhard and Shahrasbi, Amirbehshad},
  year          = 2021,
  month         = oct,
  journal       = {Journal of the ACM},
  volume        = {68},
  number        = {5},
  pages         = {1--39},
  issn          = {0004-5411, 1557-735X},
  doi           = {10.1145/3468265},
  urldate       = {2025-10-29},
  langid        = {english}
}

@misc{Aaronson_blog_AT_safety,
  author        = {Scott Aaronson},
  title         = {{My AI Safety Lecture for UT Effective Altruism}},
  journal       = {Shtetl-Optimized},
  month         = {November},
  year          = {2022},
  howpublished  = {\url{https://scottaaronson.blog/?p=6823}}
}

@inproceedings{CGZ24,
  author        = {Miranda Christ and Sam Gunn and Or Zamir},
  editor        = {Shipra Agrawal and Aaron Roth},
  title         = {Undetectable Watermarks for Language Models},
  booktitle     = {The Thirty Seventh Annual Conference on Learning Theory, June 30 - July 3, 2023, Edmonton, Canada},
  series        = {Proceedings of Machine Learning Research},
  volume        = {247},
  pages         = {1125--1139},
  publisher     = {{PMLR}},
  year          = {2024},
  url           = {https://proceedings.mlr.press/v247/christ24a.html},
  bibsource     = {dblp computer science bibliography, https://dblp.org}
}

@inproceedings{AAC+24_stoc25,
  author        = {Omar Alrabiah and Prabhanjan Ananth and Miranda Christ and Yevgeniy Dodis and Sam Gunn},
  editor        = {Michal Kouck{\'{y}} and Nikhil Bansal},
  title         = {Ideal Pseudorandom Codes},
  booktitle     = {Proceedings of the 57th Annual {ACM} Symposium on Theory of Computing, {STOC} 2025, Prague, Czechia, June 23-27, 2025},
  pages         = {1638--1647},
  publisher     = {{ACM}},
  year          = {2025},
  url           = {https://doi.org/10.1145/3717823.3718309},
  doi           = {10.1145/3717823.3718309},
  bibsource     = {dblp computer science bibliography, https://dblp.org}
}

@inproceedings{garg2025black,
  title         = {Black-box crypto is useless for pseudorandom codes},
  author        = {Garg, Sanjam and Gunn, Sam and Wang, Mingyuan},
  booktitle     = {Theory of Cryptography Conference},
  pages         = {205--224},
  year          = {2025},
  organization  = {Springer}
}

@inproceedings{dottling2025separating,
  title         = {Separating pseudorandom codes from local oracles},
  author        = {D{\"o}ttling, Nico and M{\"u}ller, Anne and Rajasree, Mahesh Sreekumar},
  booktitle     = {Theory of Cryptography Conference},
  pages         = {225--257},
  year          = {2025},
  organization  = {Springer}
}

@article{kuditipudi2024robust,
  title         = {Robust Distortion-free Watermarks for Language Models},
  author        = {Rohith Kuditipudi and John Thickstun and Tatsunori Hashimoto and Percy Liang},
  journal       = {Transactions on Machine Learning Research},
  issn          = {2835-8856},
  year          = {2024},
  url           = {https://openreview.net/forum?id=FpaCL1MO2C},
  note          = {}
}

@inproceedings{zhao2024provable,
  title         = {Provable Robust Watermarking for {AI}-Generated Text},
  author        = {Xuandong Zhao and Prabhanjan Vijendra Ananth and Lei Li and Yu-Xiang Wang},
  booktitle     = {The Twelfth International Conference on Learning Representations},
  year          = {2024},
  url           = {https://openreview.net/forum?id=SsmT8aO45L}
}

@inproceedings{kirchenbauer2023watermark,
  title         = {A watermark for large language models},
  author        = {Kirchenbauer, John and Geiping, Jonas and Wen, Yuxin and Katz, Jonathan and Miers, Ian and Goldstein, Tom},
  booktitle     = {International conference on machine learning},
  pages         = {17061--17084},
  year          = {2023},
  organization  = {PMLR}
}

@inproceedings{gunn2025an,
  title         = {An Undetectable Watermark for Generative Image Models},
  author        = {Sam Gunn and Xuandong Zhao and Dawn Song},
  booktitle     = {The Thirteenth International Conference on Learning Representations},
  year          = {2025},
  url           = {https://openreview.net/forum?id=jlhBFm7T2J}
}

@inproceedings{cohen2025watermarking,
  title         = {Watermarking language models for many adaptive users},
  author        = {Cohen, Aloni and Hoover, Alexander and Schoenbach, Gabe},
  booktitle     = {2025 IEEE Symposium on Security and Privacy (SP)},
  pages         = {2583--2601},
  year          = {2025},
  organization  = {IEEE}
}

@inproceedings{kirchenbauer2024on,
  title         = {On the Reliability of Watermarks for Large Language Models},
  author        = {John Kirchenbauer and Jonas Geiping and Yuxin Wen and Manli Shu and Khalid Saifullah and Kezhi Kong and Kasun Fernando and Aniruddha Saha and Micah Goldblum and Tom Goldstein},
  booktitle     = {The Twelfth International Conference on Learning Representations},
  year          = {2024},
  url           = {https://openreview.net/forum?id=DEJIDCmWOz}
}

@article{krishna2023paraphrasing,
  title         = {Paraphrasing evades detectors of ai-generated text, but retrieval is an effective defense},
  author        = {Krishna, Kalpesh and Song, Yixiao and Karpinska, Marzena and Wieting, John and Iyyer, Mohit},
  journal       = {Advances in neural information processing systems},
  volume        = {36},
  pages         = {27469--27500},
  year          = {2023}
}

@article{sadasivan2023can,
  title         = {Can AI-generated text be reliably detected?},
  author        = {Sadasivan, Vinu Sankar and Kumar, Aounon and Balasubramanian, Sriram and Wang, Wenxiao and Feizi, Soheil},
  journal       = {arXiv preprint arXiv:2303.11156},
  year          = {2023}
}

@inproceedings{qu2025provably,
  title         = {Provably robust multi-bit watermarking for $\{$AI-generated$\}$ text},
  author        = {Qu, Wenjie and Zheng, Wengrui and Tao, Tianyang and Yin, Dong and Jiang, Yanze and Tian, Zhihua and Zou, Wei and Jia, Jinyuan and Zhang, Jiaheng},
  booktitle     = {34th USENIX Security Symposium (USENIX Security 25)},
  pages         = {201--220},
  year          = {2025}
}

@article{kuhn2024measuring,
  title         = {Measuring the accuracy of automatic speech recognition solutions},
  author        = {Kuhn, Korbinian and Kersken, Verena and Reuter, Benedikt and Egger, Niklas and Zimmermann, Gottfried},
  journal       = {ACM Transactions on Accessible Computing},
  volume        = {16},
  number        = {4},
  pages         = {1--23},
  year          = {2024},
  publisher     = {ACM New York, NY}
}

@article{shu2024visual,
  title         = {Visual text meets low-level vision: A comprehensive survey on visual text processing},
  author        = {Shu, Yan and Zeng, Weichao and Li, Zhenhang and Zhao, Fangmin and Zhou, Yu},
  journal       = {arXiv preprint arXiv:2402.03082},
  year          = {2024}
}

@inproceedings{awasthi2019parallel,
  title         = {Parallel iterative edit models for local sequence transduction},
  author        = {Awasthi, Abhijeet and Sarawagi, Sunita and Goyal, Rasna and Ghosh, Sabyasachi and Piratla, Vihari},
  booktitle     = {Proceedings of the 2019 conference on empirical methods in natural language processing and the 9th international joint conference on natural language processing (EMNLP-IJCNLP)},
  pages         = {4260--4270},
  year          = {2019}
}

@article{fairoze2025publicly,
  title         = {Publicly-Detectable Watermarking for Language Models},
  author        = {Fairoze, Jaiden and Garg, Sanjam and Jha, Somesh and Mahloujifar, Saeed and Mahmoody, Mohammad and Wang, Mingyuan},
  journal       = {IACR Communications in Cryptology},
  volume        = {1},
  number        = {4},
  year          = {2025}
}

@inproceedings{guruswami2016efficiently,
  title         = {Efficiently decodable insertion/deletion codes for high-noise and high-rate regimes},
  author        = {Guruswami, Venkatesan and Li, Ray},
  booktitle     = {2016 IEEE International Symposium on Information Theory (ISIT)},
  pages         = {620--624},
  year          = {2016},
  organization  = {IEEE}
}

@article{guruswami2017deletion,
  title         = {Deletion codes in the high-noise and high-rate regimes},
  author        = {Guruswami, Venkatesan and Wang, Carol},
  journal       = {IEEE Transactions on Information Theory},
  volume        = {63},
  number        = {4},
  pages         = {1961--1970},
  year          = {2017},
  publisher     = {IEEE}
}

@inproceedings{cheng2018deterministic,
  title         = {Deterministic document exchange protocols, and almost optimal binary codes for edit errors},
  author        = {Cheng, Kuan and Jin, Zhengzhong and Li, Xin and Wu, Ke},
  booktitle     = {2018 IEEE 59th Annual Symposium on Foundations of Computer Science (FOCS)},
  pages         = {200--211},
  year          = {2018},
  organization  = {IEEE}
}

@inproceedings{zhao2025sok,
  title         = {Sok: Watermarking for ai-generated content},
  author        = {Zhao, Xuandong and Gunn, Sam and Christ, Miranda and Fairoze, Jaiden and Fabrega, Andres and Carlini, Nicholas and Garg, Sanjam and Hong, Sanghyun and Nasr, Milad and Tramer, Florian and others},
  booktitle     = {2025 IEEE Symposium on Security and Privacy (SP)},
  pages         = {2621--2639},
  year          = {2025},
  organization  = {IEEE}
}
\end{document}